\documentclass[jmp,amsmath,amssymb]{revtex4-2}
\usepackage[T1]{fontenc}
\usepackage[latin9]{inputenc}
\usepackage{color}
\usepackage{mathrsfs}
\usepackage{amsmath}
\usepackage{amsthm}
\usepackage{amssymb}
\usepackage{graphicx}
\usepackage{esint}
\usepackage[bookmarks=true,bookmarksnumbered=true,bookmarksopen=true,bookmarksopenlevel=1,
 breaklinks=true,pdfborder={0 0 1},backref=section,colorlinks=true]
 {hyperref}
\hypersetup{
 bookmarksdepth=3,linkcolor=blue,citecolor=blue,urlcolor=blue}

\makeatletter
\theoremstyle{plain}
\newtheorem{thm}{\protect\theoremname}
\theoremstyle{remark}
\newtheorem{rem}[thm]{\protect\remarkname}
\theoremstyle{plain}
\newtheorem{lem}[thm]{\protect\lemmaname}

\usepackage{color}
\usepackage{mathrsfs}

\AtBeginDocument{%
  \hypersetup{
    colorlinks=true,
    linkcolor=blue,
    citecolor=blue,
    urlcolor=blue,
    bookmarks=true,
    bookmarksnumbered=true,
    bookmarksopen=true,
    bookmarksdepth=3,
    pdftitle={Factorized dispersion relations for two coupled systems},
    pdfauthor={Alexander Figotin},
    pdfsubject={Dispersion relations, coupled systems, Lagrangian field theory}
  }%
}

\providecommand{\definitionname}{Definition}
\providecommand{\lemmaname}{Lemma}
\providecommand{\remarkname}{Remark}
\providecommand{\theoremname}{Theorem}

\numberwithin{figure}{section}
\numberwithin{equation}{section}
\numberwithin{table}{section}
\@ifundefined{thm}{%
  \theoremstyle{plain}%
  \newtheorem{thm}{\protect\theoremname}%
}{}
\@ifundefined{rem}{%
  \theoremstyle{remark}%
  \newtheorem{rem}[thm]{\protect\remarkname}%
}{}
\@ifundefined{lem}{%
  \theoremstyle{plain}%
  \newtheorem{lem}[thm]{\protect\lemmaname}%
}{}
\@ifundefined{defn}{%
  \theoremstyle{definition}%
}{}

\makeatother

\providecommand{\lemmaname}{Lemma}
\providecommand{\remarkname}{Remark}
\providecommand{\theoremname}{Theorem}

\begin{document}
\title{Factorized dispersion relations for two coupled systems}
\author{Alexander Figotin}
\email{afigotin@uci.edu}

\affiliation{Department of Mathematics, University of California, Irvine, CA 92697,
USA}
\begin{abstract}
We establish that the dispersion relations of any physical system
composed of two coupled subsystems, governed by a space-time homogeneous
Lagrangian, admit a factorized form $G_{1}G_{2}=\gamma G_{\mathrm{c}}$,
where $G_{1}$ and $G_{2}$ are the subsystem dispersion functions,
$G_{\mathrm{c}}$ is the coupling function, and $\gamma$ is the coupling
parameter. The result follows from a determinant expansion theorem
applied to the block structure of the coupled system matrix, and is
illustrated through three examples: the traveling wave tube, vibrations
of an airplane wing, and the Mindlin-Reissner plate theory. For the
Mindlin-Reissner example we carry out a complete asymptotic analysis
of the coupled dispersion branches, establishing that the factorized
form provides a precise quantitative measure of mode hybridization:
all four branches carry the imprint of both subsystem factors for
any nonzero coupling, while asymptotically recovering the identity
of pure uncoupled modes at large frequencies and wavenumbers. We further
analyze the universal local geometry of the coupled dispersion branches
near their intersection --- the cross-point model --- showing it
is generically hyperbolic, and present a mechanical analog in which
the wavenumber is replaced by a scalar parameter, exhibiting the same
factorized structure and avoided crossing. 
\end{abstract}
\maketitle

\section{Introduction}

Dispersion relations --- equations relating frequency $\omega$ to
wavenumber $k$ --- are among the most fundamental objects in wave
physics. They encode the propagation properties of a physical system
and govern phenomena as varied as wave packet spreading, group velocity,
band gaps, and instabilities. The structure of dispersion relations
for systems composed of two interacting subsystems is therefore of
broad physical interest. Classical treatments of wave propagation
in elastic media \citet{Achen}, \citet{GerRix}, dispersive waves
\citet{Whith1}, \citet{Whith2}, and coupled-mode theory \citet{HausHua}
provide the conceptual background for the present work.

In our recent work \citet{FigFDT1} on the field theory of traveling
wave tubes (TWT) we discovered a physically appealing factorized form
of the dispersion relations, based on the fact that the TWT can be
viewed as two interacting coupled subsystems: the electron beam and
the metal wave-guiding structure containing it. The factorized form
$G_{1}G_{2}=\gamma G_{\mathrm{c}}$ expresses the dispersion relation
of the coupled system as a product of the two subsystem dispersion
functions perturbed by the coupling, where $G_{1}$ and $G_{2}$ are
the dispersion functions of the first and second subsystems respectively,
$G_{\mathrm{c}}$ is the coupling function, and $\gamma$ is the coupling
parameter. Since decomposition into two interacting subsystems arises
in a wide variety of physical contexts, it is natural to ask whether
this factorized structure is a general property of two-subsystem Lagrangian
field theories. The answer is affirmative, and the present paper establishes
this in full generality. A striking consequence of this factorized
structure, developed in detail for the Mindlin-Reissner plate example,
is that it provides a precise and quantitative account of mode hybridization
induced by coupling: every branch of the coupled dispersion relation
carries the imprint of both subsystem factors for any nonzero coupling,
and the degree of mixing is directly controlled by the coupling parameter.

The Lagrangian framework provides the natural setting for this investigation.
Physical systems furnished with Lagrangians depending on fields and
their partial derivatives over space-time, under the assumption of
space-time homogeneity, possess dispersion relations through the Fourier-domain
eigenvalue condition $\det\{\hat{\mathsf{L}}(k)\}=0$. When the system
Lagrangian admits a decomposition into two subsystem Lagrangians coupled
by a single coupling parameter $b$, the determinant condition factors
in a precise algebraic sense governed by Theorem \ref{thm:detAB}
below, which is based on Markus's determinant formula \citet{Markus}.
We note that related asymptotic approaches to factorizing plate dispersion
relations have been developed by Kaplunov and collaborators \citet{KapNolRog},
\citet{KapNob}, \citet{ChebKapRog}, \citet{AlzKapPri}, where polynomial
approximations of the Rayleigh-Lamb and Mindlin plate equations isolate
individual wave branches; the present approach differs in that the
factorization is exact and algebraic, derived directly from the Lagrangian
coupling structure. The avoided-crossing and cross-point phenomena
that arise in the coupled dispersion branches are closely related
to the non-crossing rule of quantum mechanics \citet{Noh} and to
coupled-mode theory \citet{HausHua}.

For physical systems that do not possess dispersion relations, a natural
substitute is the dependence of the system eigenfrequencies on a physical
parameter $p$, which we refer to as \emph{frequency-parameter relations}.
We show that such relations admit the same factorized structure for
two coupled systems, as illustrated explicitly in Section~\ref{subsec:mech-crpo}.

The paper is organized as follows. Section~\ref{sec:LagHi} provides
a concise review of the Lagrangian variational framework for fields
with higher-order derivatives and the associated dispersion relations,
establishing the notation used throughout. Section~\ref{sec:factdisp}
introduces the coupled two-subsystem framework, defines the coupling
parameter, and develops the factorized form of the dispersion relation,
with the key result given by Theorem~\ref{thm:detAB} on the determinant
of the coupled system matrix. Section~\ref{sec:examp} illustrates
the theory through three physically appealing examples: the traveling
wave tube (Section~\ref{subsec:twt}), vibrations of an airplane
wing (Section~\ref{subsec:wing}), and the Mindlin-Reissner plate
theory (Section~\ref{subsec:midreis}), including a detailed asymptotic
analysis of mode hybridization and a comparison with the classical
Kirchhoff plate theory. Section~\ref{sec:cross-po-model} introduces
the cross-point model as the universal local description of two coupled
dispersion branches near their intersection, derives the associated
hyperbolic geometry, constructs the Lagrangian underlying the cross-point
dispersion relation, presents a finite-dimensional mechanical analog
in which the wavenumber is replaced by a scalar parameter, and demonstrates
that hybridization is spatially concentrated near the cross-point
with the coupled branches recovering their individual mode character
asymptotically at large frequencies and wavenumbers. The Appendix
(Section~\ref{sec:append}) collects auxiliary material including
the Fourier transform conventions and the Markus determinant formula
used in the proofs.

\section{Review of Lagrangians with higher derivatives and the dispersion
relations}

\label{sec:LagHi}

We provide here a concise review of the Lagrangian variational framework
that involves higher derivatives following mostly \citet[Sec. 11]{GelFom}
and \citet[Chap. 1, Sec. 5,6]{GiaqHild} (see also \citet[Sec. 18]{Carath},
\citet[Sec. 33]{Hass}). Our motivation for considering Lagrangian
dependent on higher order partial derivatives is that some problems
of mechanics of continua related to bending and twisting involve the
second order derivatives, see for instance \citet[Sec. 2.1, 8.7, 8.8]{Langh}.

\subsection{The Lagrangian and the Euler equations}

\label{subsec:LagEu}

Suppose that conceivable configurations of the physical system are
described by a set of real-valued fields $u_{i}\left(x\right)$, $1\leq i\leq N$,
over space-time $\mathbb{R}^{n+1}$, that is 
\begin{equation}
u\left(x\right)=\left(u_{1}\left(x\right),\ldots,u_{N}\left(x\right)\right),\quad x\in\mathbb{R}^{n+1},\label{eq:varLag1a}
\end{equation}
where $N$ is the total number of field variables. In the cases of
interest the space-time vector $x\in\mathbb{R}^{n+1}$ represents
the space $\mathbb{R}^{n}$ so that 
\begin{equation}
x=\left(x_{0},x_{1},\cdots,x_{n}\right),\quad x_{0}=ct,\label{eq:varLag1b}
\end{equation}
where $x_{1},\cdots,x_{n}$ are Cartesian spatial coordinates, $t$
is time and constant $c$ is a ``natural'' to the system velocity.
In this setting $x_{0}$ represents the time variable. In the most
of physical applications of interest $n=1,2,3$.

We assume further that the physical system is furnished with its Lagrangian
$L$. Commonly the system Lagrangian $L$ depends on the relevant
fields $u_{i}\left(x\right)$ and their first order partial derivatives,
which are 
\begin{equation}
\partial_{j}u_{i}\left(x\right),\quad\partial_{j}=\frac{\partial}{\partial x_{j}},\quad0\leq j\leq n,\quad1\leq i\leq N.\label{eq:varLag1c}
\end{equation}
In mechanics of continua though the system Lagrangian $L$ may depend
may also depend on the partial derivatives of the second order. The
physical origins of the presence of the second order derivatives are
often bending and twisting, \citet[Sec. 2.1, 8.7, 8.8]{Langh}. So
to cover all cases of interest we allow the system Lagrangians to
be dependent on the higher order partial derivatives.

To deal with system Lagrangians that may be dependent on the higher
order partial derivatives we introduce notations for them using\emph{
multi-indices} $\mu$: 
\begin{gather}
\mu=\left(\mu_{0},\mu_{1},\ldots,\mu_{n}\right),\quad\mu_{j}=0,1,2,\ldots,\quad0\leq j\leq n;\quad\left|\mu\right|=\mu_{0}+\mu_{1}+\cdots+\mu_{n}.\label{eq:varLag1d}
\end{gather}
The corresponding partial derivatives $\partial_{\mu}$ then are defined
as follows: 
\begin{equation}
\partial_{\mu}\stackrel{\mathrm{def}}{=}\frac{\partial^{\left|\mu\right|}}{\partial x_{1}^{\mu_{1}}\cdots\partial x_{n}^{\mu_{n}}}=\prod_{j=0}^{n}\partial_{j}^{\mu_{j}},\label{eq:Fvarux1a}
\end{equation}
where $\left|\mu\right|$is the\emph{ order of the partial derivative}.
Note that for $\mu=\mathbf{0}=\left(0,\ldots,0\right)$ we have $\partial_{\mathbf{0}}u=u$.
Let us denote by $I_{L}$ the set of all multi-indices $\mu$ such
that the Lagrangian $L$ in equation (\ref{eq:Fvarux1a}) depends
on $\partial_{\mu}u_{i}\left(x\right)$ for at least one $i$.

Following the general variational setup procedure we consider a set
of real-valued variables

\begin{equation}
\left\{ v_{i\mu}\right\} \stackrel{\mathrm{def}}{=}\left\{ v_{i\mu}:i=1,\ldots,N,\;\mu\in I_{L}\right\} ,\label{eq:Fvarux1b}
\end{equation}
with index structure matching exactly the same for the partial derivatives
$\partial_{\mu}u_{i}\left(x\right)$. We introduce then: (i) the system
real-valued Lagrangian function $L\left(x,\left\{ v_{i\mu}\right\} \right)$
assuming that it is infinitely differentiable with respect to variables
$\left\{ v_{i\mu}\right\} $ (often it is just a polynomial function
of relevant variables); (ii) the corresponding \emph{action integral
$\mathscr{L}$} using substitution $v_{i\mu}=\partial_{\mu}u_{i}\left(x\right)$
in the \emph{Lagrangian function} $L\left(x,\left\{ v_{i\mu}\right\} \right)$:
\begin{equation}
\mathscr{L}\left(u\right)=\int_{\Omega}L\left(x,\left\{ \partial_{\mu}u_{i}\left(x\right)\right\} \right)\,\mathrm{d}x,\label{eq:Fvarux2a}
\end{equation}
where $\Omega\subseteq\mathbb{R}^{n+1}$ is an open domain in $\mathbb{R}^{n}$.
Then the extrema of the action integral $\mathscr{L}\left(u\right)$
satisfy the following \emph{Euler equations}: 
\begin{equation}
\sum_{\mu\in I_{L}}\left(-1\right)^{\left|\mu\right|}\partial_{\mu}\left[L_{v_{i\mu}}\left(x,\left\{ \partial_{\mu}u_{i}\left(x\right)\right\} \right)\right]=0,\quad i=1,\ldots,N,\label{eq:Fvarux2c}
\end{equation}
where 
\begin{equation}
L_{v_{i\mu}}\stackrel{\mathrm{def}}{=}\partial_{v_{i\mu}}L,\quad i=1,\ldots,N,\quad\mu\in I_{L}.\label{eq:Fvarux2b}
\end{equation}

A particularly important case when the Euler-Lagrange equations are
linear, that is when the Lagrangian $L$ is a quadratic function of
$\left\{ v_{i\mu}\right\} $, namely 
\begin{equation}
L\left(x,\left\{ v_{i\mu}\right\} \right)=\frac{1}{2}\sum_{i,j=1}^{N}\sum_{\mu,\eta\in I_{L}}a_{i\mu;j\eta}\left(x\right)v_{i\mu}v_{j\eta},\label{eq:Fvarux2d}
\end{equation}
where we may assume without loss of generality that 
\begin{equation}
a_{i\mu;j\eta}\left(x\right)=a_{j\eta;i\mu}\left(x\right),\quad i,j=1,\ldots N,\quad\mu,\eta\in I_{L}.\label{eq:Fvarux2e}
\end{equation}
In this case the partial derivatives $L_{v_{i\mu}}$ defined by equations
(\ref{eq:Fvarux2b}) take the form 
\begin{equation}
L_{v_{i\mu}}=\sum_{j=1}^{N}\sum_{\eta\in I_{L}}a_{i\mu;j\eta}\left(x\right)v_{j\eta},\label{eq:Fvarux3a}
\end{equation}
and then the corresponding Euler equations (\ref{eq:Fvarux2c}) turn
into 
\begin{equation}
\sum_{\eta\in I_{L}}\left(-1\right)^{\left|\eta\right|}\partial_{\eta}\left[\sum_{j=1}^{N}\sum_{\gamma\in I_{L}}a_{i\eta;j\gamma}\left(x\right)\partial_{\gamma}u_{j}\left(x\right)\right]=0,\quad i=1,\ldots,N.\label{eq:Fvarux3b}
\end{equation}
In the case when the system is time and space homogeneous with coefficients
$a_{i\eta;j\gamma}\left(x\right)=a_{i\eta;j\gamma}$ being independent
of $x$ constants the Lagrangian (\ref{eq:Fvarux3a}) and the Euler
equations (\ref{eq:Fvarux3b}) turn into the following respective
equations: 
\begin{equation}
L=L\left(\left\{ v_{i\mu}\right\} \right)=\frac{1}{2}\sum_{i,j=1}^{N}\sum_{\mu,\eta\in I_{L}}a_{i\mu;j\eta}v_{i\mu}v_{j\eta},\label{eq:Fvarux3c}
\end{equation}
\begin{equation}
\sum_{j=1}^{N}\mathsf{L}_{ij}\left(\partial\right)u_{j}\left(x\right)=0,\quad\mathsf{L}_{ij}\left(\partial\right)\stackrel{\mathrm{def}}{=}\sum_{\eta,\gamma\in I_{L}}\left(-1\right)^{\left|\eta\right|}a_{i\eta;j\gamma}\partial_{\gamma+\eta},\quad i=1,\ldots,N.\label{eq:Fvarux3d}
\end{equation}
Note that the Euler equations (\ref{eq:Fvarux3d}) are linear as a
consequence of the quadratic dependence of the Lagrangian $L$ on
$\left\{ v_{i\mu}\right\} $ according to equations (\ref{eq:Fvarux3c}).

The Euler equations (\ref{eq:Fvarux3d}) can be also recast into the
following matrix form: 
\begin{equation}
\mathsf{L}\left(\partial\right)u\left(x\right)=0,\quad\mathsf{L}\left(\partial\right)=\left\{ \mathsf{L}_{ij}\left(\partial\right)\right\} _{i,j=1,\ldots,N},\quad u\left(x\right)=[u_{1}\left(x\right),\cdots,u_{N}\left(x\right)]^{\mathrm{T}},\label{eq:Fvarux4a}
\end{equation}
where $L\left(\partial\right)$ is $N\times N$ matrix with each entry
being differential operator $\mathrm{L}_{ij}\left(\partial\right)$
defined by equations (\ref{eq:Fvarux3d}) and $u\left(x\right)$ is
a column vector.

\subsection{The dispersion relations}

\label{subsec:disp}

To use the well-known approach for analyzing time and space homogeneous
systems we consider vector functions $u\left(x\right)$ of the form
\begin{equation}
u\left(x\right)=\exp\left\{ \mathrm{i}\left(k_{0}x_{0}-\sum_{j=1}^{n}k_{j}x_{j}\right)\right\} \hat{u}\left(k\right),\quad k=\left(k_{0},k_{1},\cdots,k_{n}\right)\in\mathbb{R}^{n+1},\quad k_{0}=\frac{\omega}{c}.\label{eq:Fvarux4b}
\end{equation}
Substituting the above form vector function $u\left(x\right)$ into
equation (\ref{eq:Fvarux4a}) we obtain the following equation: 
\begin{equation}
\hat{\mathsf{L}}\left(k\right)\hat{u}\left(k\right)=0,\quad k=\left(k_{0},k_{1},\cdots,k_{n}\right)\in\mathbb{R}^{n+1},\quad k_{0}=\frac{\omega}{c},\label{eq:Fvarux4c}
\end{equation}
where $\hat{\mathsf{L}}\left(k\right)=\left\{ \hat{\mathsf{L}}_{ij}\left(k\right)\right\} $
is $N\times N$ matrix defined by 
\begin{equation}
\hat{\mathsf{L}}\left(k\right)=\left[\hat{\mathsf{L}}_{ij}\left(k\right)\right]_{i,j=1,\ldots,N},\quad k=\left(k_{0},k_{1},\cdots,k_{n}\right)\in\mathbb{R}^{n+1}\label{eq:Fvarux5a}
\end{equation}
\begin{equation}
\hat{\mathsf{L}}_{ij}\left(k\right)\stackrel{\mathrm{def}}{=}\sum_{\eta,\gamma\in I_{L}}\left(-1\right)^{\left|\eta\right|}a_{i\eta;j\gamma}\left(-1\right)^{k_{0}}\left(-\mathrm{i}\right)^{\left|k\right|}k^{\gamma+\eta},\quad i=1,\ldots,N,\label{eq:Fvarux5b}
\end{equation}
\begin{equation}
k^{\mu}\stackrel{\mathrm{def}}{=}\prod_{j=0}^{n}k_{j}^{\mu_{j}},\quad\left|k\right|=k_{0}+k_{1}+\cdots+k_{n}.\label{eq:Fvarux5c}
\end{equation}
Note that equation (\ref{eq:Fvarux4c}) can be viewed as a generalized
eigenvalue problem with $k_{0}=\frac{\omega}{c}$ being an eigenvalue
and nontrivial $\hat{u}\left(k\right)$ being a generalized eigenvector.
According to a well known statement from the linear algebra vector
equation (\ref{eq:Fvarux4c}) has a nontrivial (nonzero) solution
$\hat{u}\left(k\right)\neq0$ if and only if 
\begin{equation}
\det\left\{ \hat{\mathsf{L}}\left(k\right)\right\} =0,\quad k=\left(k_{0},k_{1},\cdots,k_{n}\right)\in\mathbb{R}^{n+1},\quad k_{0}=\frac{\omega}{c},\label{eq:Fvarux6a}
\end{equation}
and this equation can be viewed as the \emph{dispersion relation}
between $k_{0}=\frac{\omega}{c}$ and $\bar{k}=\left(k_{1},\cdots,k_{n}\right)\in\mathbb{R}^{n}$.
Note also that equation (\ref{eq:Fvarux6a}) relates the angular frequency
$\omega$ to angular wavevector $\bar{k}$ and it is a justification
for calling it the dispersion relation.

\section{Coupled systems and the factorized form of the dispersion relations}

\label{sec:factdisp}

Quite often a decomposition of given system $S$ into say two interacting
(coupled) subsystems $S_{1}$ and $S_{2}$ is rather clear based on
physical grounds. Nevertheless there could be alternative mathematical
formulations of a basis for such a decomposition. One physically sound
approach to the system decomposition is based on the system Lagrangian
assuming that it is available. Being given such a Lagrangian we first
split the relevant fields describing the system $S$ configuration
into two groups say 
\begin{equation}
U_{1}\left(x\right)=\left\{ u_{i}\left(x\right):1\leq i\leq m<N\right\} ,\quad U_{2}\left(x\right)=\left\{ u_{i}\left(x\right):1+m\leq i\leq N\right\} .\label{eq:UUcoup1a}
\end{equation}
We associate then fields $U_{1}\left(x\right)$ and $U_{2}\left(x\right)$
with respectively subsystems $S_{1}$ and $S_{2}$ and consider the
system Lagrangian decomposition into the sum 
\begin{equation}
L=L_{1}\left\{ \partial_{\mu}U_{1}\left(x\right)\right\} +L_{2}\left\{ \partial_{\mu}U_{2}\left(x\right)\right\} +L_{12},\label{eq:UUcoup1b}
\end{equation}
where Lagrangians $L_{1}$ and $L_{2}$ represent respectively subsystems
$S_{1}$ and $S_{2}$ and $L_{12}=L-L_{1}-L_{2}$ represents the interaction
between subsystems $S_{1}$ and $S_{2}$. We expect then that there
exists a parameter $b$ of the system $S$ that is involved in the
interaction Lagrangian $L_{12}$ so that if $b=0$ then $L_{12}=0$.
If that is the case we refer to such a parameter $b$ as a \emph{coupling
parameter.}

If we can not identify the desired coupling parameter $b$ but still
insist on having subsystems $S_{1}$ and $S_{2}$ as a basis of the
system $S$ decomposition we set up a family of Lagrangians 
\begin{equation}
L^{\left(b\right)}=L_{1}\left\{ \partial_{\mu}U_{1}\left(x\right)\right\} +L_{2}\left\{ \partial_{\mu}U_{2}\left(x\right)\right\} +bL_{12},\label{eq:UUcoup1c}
\end{equation}
where Lagrangians $L_{1}$, $L_{2}$ and $L_{12}$ are the same as
defined above. In other words we introduced an additional parameter
$b$ into the system Lagrangian. Then evidently the original Lagrangian
is recovered for $b=1$, and $b$ can be viewed as a coupling parameter
since for $b=0$ we have 
\begin{equation}
L^{\left(0\right)}=L_{1}\left\{ \partial_{\mu}U_{1}\left(x\right)\right\} +L_{2}\left\{ \partial_{\mu}U_{2}\left(x\right)\right\} ,\label{eq:UUcoup1d}
\end{equation}
indicating that the subsystems $S_{1}$ and $S_{2}$ are decoupled.

Hence without loss of generality we may assume that we always have
a coupling parameter associated with the decomposition of a given
system $S$ into say two coupled subsystems $S_{1}$ and $S_{2}$.
Having a natural to the system $S$ coupling parameter that can be
controlled by us is physically preferable of course. In this case
we can control the level of coupling experimentally.

\subsection{Setting up the coupled system}

\label{subsec:coupsys}

Rather often the physical systems at hand can be naturally decomposed
into two interacting, coupled subsystems. This situation can be specified
and quantified as follows.

Based on our studies of TWT systems and the factorized form of the
relevant dispersion relations we introduce here a general model for
factorized dispersion relations of a system composed of some two coupled
systems.

The dispersion relations emerge when we recast the original homogeneous
problem in the frequency-wavevector domain with $\omega$ being the
frequency and $k\in\mathbb{R}^{n}$ being the wavevector.

Suppose we have two initially non-interacting systems. Suppose also
that the systems are governed by linear evolution equations and such
that the corresponding eigenvalue problem for each of them can be
written in the following form: 
\begin{equation}
\varLambda_{j}Q_{j}=0,\quad\varLambda_{j}=\varLambda_{j}\left(k,\omega\right)\quad j=1,2,\label{eq:Lamj1a}
\end{equation}
where $\varLambda_{j}$ is a $n_{j}\times n_{j}$ square matrix and
$Q_{j}$ is a $n_{j}$ dimensional column vector for $j=1,2$. Then
the systems dispersion relations are 
\begin{equation}
\det\left\{ \varLambda_{j}\left(k,\omega\right)\right\} =0,\quad j=1,2.\label{eq:Lamj1b}
\end{equation}
Assume now that the two systems interact and the system composed of
these interacting subsystems is described by the following linear
problem: 
\begin{gather}
AQ=0,\quad A=\varLambda+B,\quad\varLambda=\varLambda\left(k,\omega\right)=\left[\begin{array}{rr}
\varLambda_{1}\left(k,\omega\right) & 0\\
0 & \varLambda_{2}\left(k,\omega\right)
\end{array}\right],\label{eq:Lamj1c}\\
B=B\left(k,\omega\right)=\left[\begin{array}{rr}
B_{11}\left(k,\omega\right) & B_{12}\left(k,\omega\right)\\
B_{21}\left(k,\omega\right) & B_{22}\left(k,\omega\right)
\end{array}\right],\quad Q=Q\left(k,\omega\right)=\left[\begin{array}{r}
Q_{1}\left(k,\omega\right)\\
Q_{2}\left(k,\omega\right)
\end{array}\right]\nonumber 
\end{gather}
where $B$ is referred to as \emph{coupling matrix} where submatrices
$B_{ij}$, $j=1,2$ may depend on $k$ and $\omega$.

It is convenient to modify the definition of coupling matrix $B$
by scaling it with a scalar real valued factor $b$. Consequently,
the eigenvalue problem (\ref{eq:Lamj1c}) turns into 
\begin{equation}
AQ=0,\quad A=A\left(b\right)=\varLambda+bB,\quad\varLambda=\left[\begin{array}{rr}
\varLambda_{1} & 0\\
0 & \varLambda_{2}
\end{array}\right],\quad B=\left[\begin{array}{rr}
B_{11} & B_{12}\\
B_{21} & B_{22}
\end{array}\right],\quad Q=\left[\begin{array}{r}
Q_{1}\\
Q_{2}
\end{array}\right].\label{eq:Lamj1d}
\end{equation}
Note then that matrix $A\left(b\right)$ defined by equations (\ref{eq:Lamj1d})
depends linearly on $b$ and $A\left(0\right)=\varLambda$ and $A\left(1\right)=\varLambda+B$.
In other words, $b=0$ corresponds to the case when subsystems are
completely decoupled and described by the matrix $\varLambda$ is
in equations (\ref{eq:Lamj1c}) whereas for $b=1$ we get the original
coupling matrix $B$ is in equations (\ref{eq:Lamj1c}).

Then the dispersion relations of the coupled system are consequently
\begin{equation}
\det\left\{ A\left(b\right)\right\} =\det\left\{ \varLambda\left(k,\omega\right)+bB\left(k,\omega\right)\right\} =0.\label{eq:Lamj1b-1}
\end{equation}

\subsection{Factorized form of the dispersion relation}

\label{subsec:fact_dispG}

The factorized dispersion relation assumes that the original system
is composed of two coupled (interacting) subsystems. Mathematical
representation of the coupling comes through a particular form of
the \emph{system matrix} $M_{k}\left(b\right)$ where $b$ is a scalar-valued
\emph{coupling coefficient}. Namely, we assume the system matrix $M_{k}\left(b\right)$
to be of the form 
\begin{equation}
M_{k}\left(b\right)=\varLambda+bB\left(b\right),\quad\varLambda=\left[\begin{array}{rr}
\varLambda_{1} & 0\\
0 & \varLambda_{2}
\end{array}\right]\quad B\left(b\right)=\left[\begin{array}{rr}
B_{11}\left(b\right) & B_{12}\left(b\right)\\
B_{21}\left(b\right) & B_{22}\left(b\right)
\end{array}\right].\label{eq:MkbAB1a}
\end{equation}
where matrix $B\left(b\right)$ is assumed to depend on $b$ polynomially.
Note that when the coupling coefficient $b=0$ then according to equations
(\ref{eq:MkbAB1a}) $M_{k}\left(0\right)=\varLambda$ where $\varLambda$
is a block-diagonal matrix. The fact that $\varLambda$ is block-diagonal
manifests the decomposition of the original system into two non-interacting
subsystems with respective system matrices $\varLambda_{1}$ and $\varLambda_{2}$.
The particular choice $bB\left(b\right)$ in equations (\ref{eq:MkbAB1a})
to represent the subsystems interaction is justified by two requirements:
(i) $B\left(b\right)$ depends on $b$ polynomially and (ii) the coupling/interaction
has to vanish as $b=0$.

We start off with the following implication of Markus's formula (\ref{eq:detAB1b})
for $\det\left\{ A+bB\left(b\right)\right\} $. 
\begin{thm}[determinant of the coupled systems matrix]
\label{thm:detAB} Let $A$ and $B\left(b\right)$ be two $n\times n$
matrices with $n\geq2$. Assume also that $b$ is a complex number
and matrix $B\left(b\right)$ depends on $b$ polynomially, that is
\begin{equation}
B\left(b\right)=\sum_{s=0}^{m}B^{\left(s\right)}b^{s},\label{eq:detAB1ab}
\end{equation}
where $m\geq0$ is an integer and $B^{\left(s\right)}$ are $n\times n$
matrices. Then $\det\left\{ A+bB\left(b\right)\right\} $ is a polynomial
function of $b$ satisfying the following representation 
\begin{equation}
\det\left\{ A+bB\left(b\right)\right\} =\det\left\{ A\right\} +\sum_{r=1}^{n-1}c_{r}\left(b\right)b^{r}+b^{n}\det\left\{ B\left(b\right)\right\} ,\label{eq:detAB1b}
\end{equation}
where 
\begin{equation}
c_{r}\left(b\right)=\sum_{\alpha,\beta\in Q_{n-r,n}}\left(-1\right)^{\left|\alpha\right|+\left|\beta\right|}\det\left\{ A\left[\alpha|\beta\right]\right\} \det\left\{ B\left(b\right)\left[\alpha^{c},\beta^{c}\right]\right\} ,\quad1\leq r\leq n-1,\quad n\geq2,\label{eq:detAB1c}
\end{equation}
Coefficient $c_{1}\left(b\right)$ satisfies the following representation
\begin{equation}
c_{1}\left(b\right)=\mathrm{tr}\,\left\{ A^{\mathrm{A}}B\left(b\right)\right\} ,\label{eq:detAB2a}
\end{equation}
where $A^{\mathrm{A}}$ is the adjugate to $A$ matrix defined by
equations (\ref{eq:adjug1a}). In the case when $A$ is a diagonal
matrix equations (\ref{eq:detAB2a}) the following representation
holds for $\mathrm{tr}\,\left\{ A^{\mathrm{A}}B\left(b\right)\right\} $
\begin{equation}
\mathrm{tr}\,\left\{ A^{\mathrm{A}}B\left(b\right)\right\} =\left[\sum_{i=1}^{n}\left(\prod_{j\neq i}A_{j,j}\right)B_{i,i}\left(b\right)\right].\label{eq:detAB2c}
\end{equation}
Equations (\ref{eq:detAB1b})-(\ref{eq:detAB2a}) readily imply 
\begin{equation}
\det\left\{ A+bB\left(b\right)\right\} =\det\left\{ A\right\} +\mathrm{tr}\,\left\{ A^{\mathrm{A}}B\left(0\right)\right\} b+O\left(b^{2}\right),\quad b\rightarrow0.\label{eq:detAB2ca}
\end{equation}
\end{thm}

\begin{proof}
Formula (\ref{eq:detAB1b}) for $\det\left\{ A+bB\left(b\right)\right\} $
follows straightforwardly from Markus's formula (\ref{eq:detApB1}).
As to equation (\ref{eq:detAB2a}) for $c_{1}\left(b\right)$ it is
verified by using (i) equation (\ref{eq:detAB1c}) for $r=1$ and
(ii) the definition (\ref{eq:adjug1a}) of adjugate matrix $A^{\mathrm{A}}$.
Finally, asymptotic formula (\ref{eq:detAB2ca}) follows readily from
equations (\ref{eq:detAB1b}) and (\ref{eq:detAB2a}). Note that for
any two $n\times n$ matrices $C$ and $D$ we have 
\begin{equation}
\mathrm{tr}\,\left\{ CD\right\} =\sum_{i=1}^{n}C_{i,j}D_{j,i}.\label{eq:detAB2d}
\end{equation}
Equation (\ref{eq:detAB2c}) follows from (i) the definition (\ref{eq:adjug1a})
of adjugate $A^{\mathrm{A}}$ applied to the special case of a diagonal
matrix $A$ and (ii) equation (\ref{eq:detAB2d}) applied for $C=A^{\mathrm{A}}$
and $D=B\left(b\right)$. 
\end{proof}
\begin{rem}[factorized dispersion relation]
\label{rem:factdisp2} Applying Theorem~\ref{thm:detAB} to the
system matrix $M_{k}\left(b\right)=\varLambda+bB\left(b\right)$ defined
in equation~(\ref{eq:MkbAB1a}), with $A=\varLambda$, yields the
dispersion relation $\det\left\{ M_{k}\left(b\right)\right\} =0$.
The key step connecting this to the factorized form $G_{1}G_{2}=\gamma G_{\mathrm{c}}$
is the observation that, since $\varLambda=\mathrm{diag}\left(\varLambda_{1},\varLambda_{2}\right)$
is block-diagonal, 
\begin{equation}
\det\left\{ \varLambda\right\} =\det\left\{ \varLambda_{1}\right\} \det\left\{ \varLambda_{2}\right\} =G_{1}G_{2},\label{eq:detLam12}
\end{equation}
where $G_{j}=\det\left\{ \varLambda_{j}\right\} $ is the dispersion
function of subsystem $S_{j}$, $j=1,2$. The remaining terms in the
expansion~(\ref{eq:detAB1b}) then play the role of $\gamma G_{\mathrm{c}}$,
so that the dispersion relation $\det\left\{ M_{k}\left(b\right)\right\} =0$
takes precisely the factorized form~(\ref{eq:GGgamG1b}). Thus Theorem~\ref{thm:detAB}
is the algebraic engine behind the factorization, and the block-diagonal
structure of $\varLambda$ is its physical driver. In particular,
the coupling parameter $b$ of the system matrix $M_{k}\left(b\right)$
plays the role of the coupling coefficient $\gamma$ in the factorized
dispersion relation~(\ref{eq:GGgamG1b}): setting $b=0$ decouples
the two subsystems and recovers $G_{1}G_{2}=0$, while increasing
$b$ from zero introduces the interaction term $\gamma G_{\mathrm{c}}$
that perturbs the product of the bare dispersion functions. 
\end{rem}

\section{Physically appealing examples}

\label{sec:examp}

We illustrate here the efficiency of our theory by implementing it
in a number of physically appealing examples.

\subsection{Traveling wave tube}

\label{subsec:twt}

The TWT-system Lagrangian $\mathcal{L}{}_{\mathrm{TB}}$ is defined
similarly to its expression in \citep[Chap. 4, 24]{FigTWTbk} with
the only difference that there is an additional term related to \emph{serial
capacitance} $C_{\mathrm{c}}$, namely 
\begin{gather}
\mathcal{L}{}_{\mathrm{TB}}=\mathcal{L}_{\mathrm{B}}+\mathcal{L}_{\mathrm{Tb}},\;\mathcal{L}_{\mathrm{B}}=\frac{1}{2\beta}\left(\partial_{t}q+\mathring{v}\partial_{z}q\right)^{2}-\frac{2\pi}{\sigma_{\mathrm{B}}}q^{2},\label{eq:LaTBq2a}\\
\mathcal{L}_{\mathrm{Tb}}=\frac{L}{2}\left(\partial_{t}Q\right)^{2}-\frac{1}{2C}\left(\partial_{z}Q+b\partial_{z}q\right)^{2}-\frac{1}{2C_{\mathrm{c}}}Q^{2},\nonumber 
\end{gather}
where $b$ is the so-called coupling constant which is a dimensionless
phenomenological parameter and other parameters are discussed in \citep[Chap. 4, 24]{FigTWTbk}.
Constant $b$ is assumed often to satisfy $0<b\leq1$ effectively
reducing the inductive input of the e-beam current into the shunt
current, see \citep[Chap. 3]{FigTWTbk} for more details. Note that
coupling between the GTL and e-beam is introduced through term $-\frac{1}{2C}\left(\partial_{z}Q+b\partial_{z}q\right)^{2}$indicating
that the GTL distributed shunt capacitance $C$ is shared with e-beam.
Following the developments in \citep[Chap. 4, 24]{FigTWTbk} we introduce
the \emph{TWT principal parameter} $\gamma$ defined by 
\begin{equation}
\gamma=\frac{b^{2}}{C}\beta=\frac{b^{2}}{C}\frac{\sigma_{\mathrm{B}}}{4\pi}\omega_{\mathrm{rp}}^{2},\quad\omega_{\mathrm{rp}}^{2}=R_{\mathrm{sc}}^{2}\frac{4\pi\mathring{n}e^{2}}{m}.\label{eq:LaBTq2b}
\end{equation}

The Euler-Lagrange (EL) equations corresponding to the Lagrangian
$\mathcal{L}{}_{\mathrm{TB}}$ defined by equations (\ref{eq:LaTBq2a})
are the following system of the second-order differential equations
\begin{gather}
L\partial_{t}^{2}Q-\frac{1}{C}\partial_{z}^{2}\left(Q+bq\right)+\frac{1}{C_{\mathrm{c}}}Q^{2}=0,\label{eq::LaTBq2c}\\
\frac{1}{\beta}\left(\partial_{t}+\mathring{v}\partial_{z}\right)^{2}q+\frac{4\pi}{\sigma_{\mathrm{B}}}q-\frac{b}{C}\partial_{z}^{2}\left(Q+bq\right)=0,\quad\beta=\frac{\sigma_{\mathrm{B}}}{4\pi}\omega_{\mathrm{rp}}^{2}.\label{eq:LaTBq2d}
\end{gather}
The Fourier transformation (see Appendix \ref{sec:four}) in time
$t$ and space variable $z$ of equations (\ref{eq::LaTBq2c}) and
(\ref{eq:LaTBq2d}) yields 
\begin{gather}
\left(\frac{k^{2}}{C}-\omega^{2}L+\frac{1}{C_{\mathrm{c}}}\right)\hat{Q}+k^{2}\frac{b}{C}\hat{q}=0,\label{eq:LaTBq3a-1}\\
\frac{bk^{2}}{C}\hat{Q}+\left\{ \frac{b^{2}k^{2}}{C}+\frac{4\pi}{\sigma_{\mathrm{B}}}\left[1-\frac{\left(\omega-\mathring{v}k\right)^{2}}{\omega_{\mathrm{rp}}^{2}}\right]\right\} \hat{q}=0,\label{eq:LaTBq3b-1}
\end{gather}
where functions $\hat{Q}=\hat{Q}\left(k,\omega\right)$ and $\hat{q}=\hat{q}\left(k,\omega\right)$
are the Fourier transforms of the system variables $Q\left(t,z\right)$
and $q\left(t,z\right)$. We will refer to equations (\ref{eq:LaTBq3a-1}),
(\ref{eq:LaTBq3b-1}) as\emph{ transformed EL equations}. The TWT-system
eigenmodes are naturally assumed to be of the form 
\begin{equation}
Q\left(z,t\right)=\hat{Q}\left(k,\omega\right)\mathrm{e}^{-\mathrm{i}\left(\omega t-kz\right)},\quad q\left(z,t\right)=\hat{q}\left(k,\omega\right)\mathrm{e}^{-\mathrm{i}\left(\omega t-kz\right)},\label{eq:LaTBq3c-1}
\end{equation}
where $\omega$ and $k=k\left(\omega\right)$ are the frequency and
the wavenumber, respectively.

Multiplying the EL equations (\ref{eq:LaTBq3a-1}), (\ref{eq:LaTBq3b-1})
by $C$ we can recast them into the following matrix form: 
\begin{equation}
M_{k\omega}x=0,\quad M_{k\omega}=\left[\begin{array}{rr}
k^{2}-\frac{\omega^{2}}{w^{2}}+\frac{C}{C_{\mathrm{c}}} & bk^{2}\\
bk^{2} & b^{2}k^{2}+\frac{4\pi C}{\sigma_{\mathrm{B}}}\left[1-\frac{\left(\omega-\mathring{v}k\right)^{2}}{\omega_{\mathrm{rp}}^{2}}\right]
\end{array}\right],\quad x=\left[\begin{array}{r}
\hat{Q}\\
\hat{q}
\end{array}\right].\label{eq:LaTBq4a}
\end{equation}
Note that equations (\ref{eq:LaTBq4a}) can be viewed as an eigenvalue
type problem for $k$ and $x$ assuming that $\omega$ and other parameters
are fixed.

Taking into account expressions 
\[
w=\frac{1}{\sqrt{CL}},\quad\omega_{\mathrm{c}}=wk_{\mathrm{c}}=\frac{1}{\sqrt{C_{\mathrm{c}}L}}
\]
for $w$ and $\omega_{\mathrm{c}}$ as well as expression (\ref{eq:LaBTq2b})
for the TWT principle parameter $\gamma$ we can rewrite equations
(\ref{eq:LaTBq4a}) as 
\begin{equation}
M_{k\omega}x=0,\quad M_{k\omega}=M_{k\omega}\left(b\right)=\left[\begin{array}{rr}
k^{2}+\frac{\omega_{\mathrm{c}}^{2}-\omega^{2}}{w^{2}} & bk^{2}\\
bk^{2} & b^{2}\left[k^{2}+\frac{\omega_{\mathrm{rp}}^{2}-\left(\omega-\mathring{v}k\right)^{2}}{\gamma}\right]
\end{array}\right],\quad x=\left[\begin{array}{r}
\hat{Q}\\
\hat{q}
\end{array}\right].\label{eq:LaTBq4b}
\end{equation}
Yet another equivalent form of equations (\ref{eq:LaTBq4b}) can be
obtained by using phase velocity $u=\frac{\omega}{k}$ instead of
wavenumber $k$ in equation (\ref{eq:LaTBq4b}), namely 
\begin{gather}
M_{u\omega}x=0,\quad M_{u\omega}=M_{u\omega}\left(b\right)=\left[\begin{array}{rr}
\frac{\omega^{2}}{u^{2}}+\frac{\omega_{\mathrm{c}}^{2}-\omega^{2}}{w^{2}} & \frac{b\omega^{2}}{u^{2}}\\
\frac{b\omega^{2}}{u^{2}} & \left[\frac{\omega^{2}}{u^{2}}+\frac{1}{\gamma}\left(\omega_{\mathrm{rp}}^{2}-\frac{\omega^{2}\left(u-\mathring{v}\right)^{2}}{u^{2}}\right)\right]b^{2}
\end{array}\right],\quad x=\left[\begin{array}{r}
\hat{Q}\\
\hat{q}
\end{array}\right],\label{eq:LaTBq5a}
\end{gather}
where we use once again the principal TWT parameter $\gamma=\frac{b^{2}}{C}\beta=\frac{b^{2}}{C}\frac{\sigma_{\mathrm{B}}}{4\pi}\omega_{\mathrm{rp}}^{2}$
defined by equations (\ref{eq:LaBTq2b}).

Note that matrices $M_{k\omega}\left(b\right)$ and $M_{u\omega}\left(b\right)$
satisfy the following factorized representation: 
\begin{equation}
M_{k\omega}\left(b\right)=D_{b}M_{k\omega}\left(1\right)D_{b},\quad M_{u\omega}\left(b\right)=D_{b}M_{u\omega}\left(1\right)D_{b}\quad D_{b}=\left[\begin{array}{rr}
1 & 0\\
0 & b
\end{array}\right],\label{eq:LaBTq5b}
\end{equation}
where matrices $M_{k\omega}\left(1\right)$ and $M_{u\omega}\left(1\right)$
evidently do not depend on $b$.

\subsection{Vibration of an airplane wing}

\label{subsec:wing}

A simplified one-dimensional model that accounts for vibrations of
an airplane is a beam with variable section properties and variable
mass distribution, \citet[Sec. 2.1, 8.6]{Langh}. When the wing vibrates,
the segment included between two neighboring cross-sectional planes
is displaced in its plane as a rigid \emph{lamina (thin layer, plate)}.
The displacement of the lamina can be described by rotation angle
$\theta$ about a chosen point $P$ and its translation motion. The
rotation $\theta$ is independent of the location of point $P$. It
is convenient to choose point $P$ to be the center of mass of the
lamina.

The wing model parameters and involved variables are as follows: 
\begin{itemize}
\item $x$ is the axial (horizontal) coordinate along the wing; 
\item $w=w\left(x\right)$ is the vertical displacement, $z$-axis of the
center of mass $P=P\left(x\right)$ of the lamina at $x$; 
\item $\theta=\theta\left(x\right)$ is the rotation (twisting, torsion)
of the lamina at $x$; 
\item $m=m\left(x\right)$ is the linear mass density of the lamina, that
is $m\,\mathrm{d}x$ is the mass of a lamina of thickness $\mathrm{d}x$; 
\item $I_{\mathrm{m}}=I_{\mathrm{m}}\left(x\right)$ is the linear density
of mass moment, that is $I_{\mathrm{m}}=I_{\mathrm{m}}\,\mathrm{d}x$
is inertia of the lamina of thickness $\mathrm{d}x$ about point $P$; 
\item $P^{\prime}=P^{\prime}\left(x\right)$ is the centroid of the lamina
defined in Remark \ref{rem:centro} below. Note that points $P$ and
$P^{\prime}$ ordinarily do not coincide; 
\item $I=I\left(x\right)$ is the linear density of the moment of inertia
of the cross-section of the structural parts of the wing about the
principal axis of inertia through point $P^{\prime}$, that is $I=I\,\mathrm{d}x$
is the corresponding moment of inertia of the lamina of thickness
$\mathrm{d}x$; 
\item $a=a\left(x\right)$ is the algebraic distance between the center
of mass $P$ and the centroid $P^{\prime}$ in the direction of $y$-axis
(orthogonal to $z$ and $x$ axes); 
\item $E$ is the Young module, $J$ is a the so called \emph{polar moment
of inertia}, $G$ is the \emph{shear modulus}, and the product $GJ$
is the called ``torsional stiffness'', \citet[Sec. 2.1, 8.6]{Langh}.
In general quantities $J$, $G$ and $GJ$ functions of $x$. For
a solid shaft of circular cross section or for a hollow shaft of annular
cross section, $J$ is the polar moment of inertia of the cross section
about its center. For any other shape of cross section, $J$ is less
than the polar moment of inertia; 
\item The vertical deflection of the centroid $P^{\prime}$ is approximately
$w+a\theta$. 
\end{itemize}
\begin{rem}[centroid]
\label{rem:centro} In mathematics and physics, the centroid, also
known as ``geometric center'' or ``center'' of figure, of a plane
figure or solid figure is the arithmetic mean position of all the
points in the figure. The same definition extends to any object in
$n$-dimensional Euclidean space. The centroid $C$ of a subset $S$
of $\mathbb{R}^{n}$ is defined as 
\begin{equation}
C=\frac{\int_{S}x\,\mathrm{d}x}{\int_{S}\,\mathrm{d}x}.\label{eq:vibwin1a}
\end{equation}
The centroid coincides with the center of mass or the center of gravity
only if the material of the body is homogeneous. A geometric centroidal
axis is an axis that passes through the centroid of a cross section.
The concept of centroid arises naturally in many areas of physics,
in particularly in fluid mechanics, namely the centroid of a body
is its buoyancy center, \citet[Sec. 2.8]{WhiteF}. 
\end{rem}

The kinetic energy density $T$ of the wing can be represented as
follows:

\begin{equation}
T=\frac{1}{2}\left[m\left(\partial_{t}w\right)^{2}+I_{m}\left(\partial_{t}\theta\right)^{2}\right].\label{eq:vibwin1b}
\end{equation}
By the elementary beam theory the strain energy density $U_{1}$ of
bending of the wing and the strain energy $U_{2}$ due to twisting
are 
\begin{equation}
U_{1}=\frac{1}{2}EI\left(\partial_{x}^{2}\left(w+a\theta\right)\right)^{2},\quad U_{2}=\frac{1}{2}GJ\left(\partial_{x}\theta\right)^{2}.\label{eq:vibwin1c}
\end{equation}
Consequently, the total strain energy density $U$ of the wing is
\begin{equation}
U=U_{1}+U_{2}=\frac{1}{2}EI\left(\partial_{x}^{2}\left(w+a\theta\right)\right)^{2}+\frac{1}{2}GJ\left(\partial_{x}\theta\right)^{2}.\label{eq:vibwin1d}
\end{equation}
In view of equations (\ref{eq:vibwin1b}) and (\ref{eq:vibwin1d})
we get the following expression for wing Lagrangian 
\begin{equation}
L=T-U=\frac{1}{2}\left[m\left(\partial_{t}w\right)^{2}+I_{m}\left(\partial_{t}\theta\right)^{2}-\frac{1}{2}EI\left(\partial_{x}^{2}\left(w+a\theta\right)\right)^{2}-\frac{1}{2}GJ\left(\partial_{x}\theta\right)^{2}\right],\label{eq:vibwin2a}
\end{equation}
where vertical displacement $w=w\left(x\right)$ and $\theta=\theta\left(x\right)$
is the rotation of the lamina the fields that determine the wing configuration.
Since we interested in the dispersion relations we have to assume
from now on that the wing parameters $\rho$, $I_{m}$, $E$, $I$,
$a$, $G$ and $J$ are constants independent of $x$.

We can clearly see from the expression (\ref{eq:vibwin2a}) for the
wing Lagrangian that the wing system is naturally composed of two
subsystems. The first subsystem depends on vertical displacement $w=w\left(x\right)$
and the second one depends on rotation $\theta=\theta\left(x\right)$.
Indeed, let us introduce a dimensionless coupling parameter $b$ and
the Lagrangian 
\begin{equation}
L^{(b)}=\frac{1}{2}\left[m\left(\partial_{t}w\right)^{2}+I_{m}\left(\partial_{t}\theta\right)^{2}-\frac{1}{2}EI\left(\partial_{x}^{2}\left(w+ba\theta\right)\right)^{2}-\frac{1}{2}GJ\left(\partial_{x}\theta\right)^{2}\right].\label{eq:vibwin2b}
\end{equation}
One can readily verify that Lagrangian $L^{(b)}$ can be decomposed
as follows: 
\begin{gather*}
L^{(b)}=L_{w}+L_{\theta}+\frac{1}{2}ba\partial_{x}^{2}\theta\left(ba\partial_{x}^{2}\theta+2\partial_{x}^{2}w\right)\\
L_{w}=\frac{1}{2}\left[m\left(\partial_{t}w\right)^{2}-\frac{1}{2}EI\left(\partial_{x}^{2}w\right)^{2}\right],\quad L_{\theta}=\frac{1}{2}\left[I_{m}\left(\partial_{t}\theta\right)^{2}-\frac{1}{2}GJ\left(\partial_{x}\theta\right)^{2}\right]
\end{gather*}

Note that 
\begin{equation}
L^{\left(1\right)}=L,\quad L^{\left(0\right)}=\frac{1}{2}\left[m\left(\partial_{t}w\right)^{2}-\frac{1}{2}EI\left(\partial_{x}^{2}w\right)^{2}\right]+\frac{1}{2}\left[I_{m}\left(\partial_{t}\theta\right)^{2}-\frac{1}{2}GJ\left(\partial_{x}\theta\right)^{2}\right].\label{eq:vibwin2c}
\end{equation}
For $b=1$ the Lagrangian $L^{(b)}$ turns into the original wing
Lagrangian $L$ and $b=0$ the Lagrangian $L^{(0)}$ evidently represents
two decoupled subsystems with Lagrangians that dependent respectively
on vertical displacement $w=w\left(x\right)$ and rotation $\theta=\theta\left(x\right)$.
This is an example of coupling between vertical displacement $w=w\left(x\right)$
and rotation $\theta=\theta\left(x\right)$.

The EL equations associated with Lagrangian (\ref{eq:vibwin2b}) are
as follows, \citet[Sec. 2.1, 8.6]{Langh}: 
\begin{gather}
m\partial_{t}^{2}w+\partial_{x}^{2}\left(EI\phi\right)=0,\quad\phi=\partial_{x}^{2}w+b\partial_{x}^{2}\left(a\theta\right),\label{eq:vibwin2d}\\
I_{m}\partial_{t}^{2}\theta+EI\partial_{x}^{2}\phi-\partial_{x}\left(GJ\partial_{x}\theta\right)-2\partial_{x}\left(EI\phi b\partial_{x}a\right)+\partial_{x}^{2}\left(EIba\phi\right)=0.\label{eq:vibwin2e}
\end{gather}
In particular in the case when all involved system parameters $E$,
$I$, $G$, $J$, $a$ and $I_{m}$ are constant the above EL equations
turn into 
\begin{gather}
m\partial_{t}^{2}w+EI\left(\partial_{x}^{4}w+ab\partial_{x}^{4}\theta\right)=0,\label{eq:vibwin3a}\\
I_{m}\partial_{t}^{2}\theta-GJ\partial_{x}^{2}\theta+EIba\left(\partial_{x}^{4}w+ba\partial_{x}^{4}\theta\right)=0.\label{eq:vibwin3b}
\end{gather}

To obtain the dispersion relations associated with the Euler-Lagrange
equations (\ref{eq:vibwin3a}), (\ref{eq:vibwin3b}) we consider the
system eigenmodes represented as follows: 
\begin{equation}
\theta\left(x,t\right)=\hat{\theta}\left(k,\omega\right)\mathrm{e}^{-\mathrm{i}\left(\omega t-kx\right)},\quad w\left(x,t\right)=\hat{w}\left(k,\omega\right)\mathrm{e}^{-\mathrm{i}\left(\omega t-kx\right)},\label{eq:vibwin3c}
\end{equation}
where $\omega$ and $k=k\left(\omega\right)$ are the frequency and
the wavenumber, respectively. The Fourier transformation (see Appendix
\ref{sec:four}) in time $t$ and space variable $x$ of the Euler-Lagrange
equations (\ref{eq:vibwin3a}), (\ref{eq:vibwin3b}) can be written
in the following matrix form: 
\begin{equation}
B_{b}X=0,\quad B_{b}=\left[\begin{array}{rr}
I_{m}\omega^{2}-\left(b^{2}a^{2}EIk^{2}+GJ\right)k^{2} & baEIk^{4}\\
baEIk^{4} & m\omega^{2}-EIk^{4}
\end{array}\right],\quad X=\left[\begin{array}{r}
\hat{\theta}\left(k,\omega\right)\\
\hat{w}\left(k,\omega\right)
\end{array}\right].\label{eq:vibwin3d}
\end{equation}
The above formula readily implies the following equation for the Taylor
series of matrix $B_{b}$ at $b=0$:

\begin{equation}
B_{b}=\left[\begin{array}{rr}
I_{m}\omega^{2}-GJk^{2} & 0\\
0 & m\omega^{2}-EIk^{4}
\end{array}\right]+baEIk^{4}\left[\begin{array}{rr}
0 & 1\\
1 & 0
\end{array}\right]-b^{2}a^{2}EIk^{4}\left[\begin{array}{rr}
1 & 0\\
0 & 0
\end{array}\right],\label{eq:vibwin3da}
\end{equation}
The dispersion relations associated with equations (\ref{eq:vibwin3d})
are 
\begin{equation}
\det\left\{ B_{b}\right\} =\left(I_{m}\omega^{2}-GJk^{2}\right)\left(m\omega^{2}-EIk^{4}\right)-b^{2}a^{2}EIk^{2}m\omega^{2}=0,\label{eq:vibwin3e}
\end{equation}
or equivalently 
\begin{equation}
\left(I_{m}\omega^{2}-GJk^{2}\right)\left(m\omega^{2}-EIk^{4}\right)=b^{2}a^{2}EIk^{2}m\omega^{2}.\label{eq:vibwin3f}
\end{equation}

\subsection{Mindlin-Reissner theory for plates}

\label{subsec:midreis}

The Mindlin-Reissner is a plate theory for rectangular and circular
plates of constant thickness, see original papers \citet{Reiss},
\citet{Mindlin} and a review paper \citet{Liew}. According to E.
Magrab the Mindlin-Reissner theory is an improved plate theory which
is ``the direct equivalent of using the Timoshenko beam theory as
an improved theory with respect to the Euler--Bernoulli beam theory.'',
\citet[Sec. 7.1]{Magrab}. We provide below a concise review of the
Mindlin-Reissner following mostly \citet[Sec. 7]{Magrab}, \citet[Sec. 12.3]{Leissa},
\citet[Sec. 14.9]{RaoVCS}, \citet[Chap. 10.1]{ReddyPS}. J. Reddy
refers to the Mindlin-Reissner plate theory as the first-order shear
deformation plate theory (FSDT), \citet[Chap. 10.1]{ReddyPS}. He
developed a more accurate third-order shear deformation plate theory
(TSDT), \citet[Chap. 10.3]{ReddyPS}.

Consider a rectangular plate of constant thickness $h$ whose top
and bottom surfaces are parallel to the $(x,y)$-plane with the coordinate
system located midway between these surfaces. The plate has a length
$a$ in the $x$-direction, a length $b$ in the $y$-direction. The
thickness $h$, which is in the $z$-direction, is such that $h\ll a$
and $h\ll b$. The plate has a \emph{density} $\rho$, a \emph{Young's
modulus} $E$ and a \emph{Poisson's ratio} $\nu$. Let $u$ and $v$
be respectively the in-plane displacements in the $x$ and $y$ directions
and $w$ be the transverse displacement in the $z$-direction. The
displacement $w$ is assumed to be independent of $z$ and the surfaces
of the plate are stress-free; that is $\sigma_{zz}=0$. As with the
Timoshenko beam, let us assume that the in-plane displacements are
proportional to the $z$ coordinate as follows, \citet[Sec. 8.3.1]{Graff},
\citet[Sec. 7.1.1]{Magrab}, 
\begin{equation}
u=z\psi_{x}\left(x,y,t\right),\quad v=z\psi_{y}\left(x,y,t\right),\quad w=w\left(x,y,t\right),\label{eq:mireco1a}
\end{equation}
where $\psi_{x}$ is the \emph{rotation of the cross section about
a line parallel to the $y$-axis} and $\psi_{y}$ is the rotation
of the cross section about a line parallel to the $x$-axis.

The Lagrangian $L$ for the Mindlin-Reissner theory is defined as
follows, \citet[Sec. 7.1.2]{Magrab}, \citet[Sec. 14.9.2]{RaoVCS},
\citet[Sec. 4.6]{Szil}: 
\begin{equation}
L=T-U,\quad T=\frac{\rho h}{2}\left[\frac{h^{2}}{12}\left(\partial_{t}\psi_{x}\right)^{2}+\frac{h^{2}}{12}\left(\partial_{t}\psi_{y}\right)^{2}+\left(\partial_{t}w\right)^{2}\right],\label{eq:mireco1b}
\end{equation}
\begin{gather}
U=\frac{D}{2}\left[\left(\partial_{x}\psi_{x}\right)^{2}+\left(\partial_{y}\psi_{y}\right)^{2}+2\nu\left(\partial_{x}\psi_{x}\right)\left(\partial_{y}\psi_{y}\right)+\frac{1-\nu}{2}\left(\partial_{y}\psi_{x}+\partial_{x}\psi_{y}\right)^{2}\right]\label{eq:mireco1c}\\
+\frac{\kappa hG}{2}\left[\left(\psi_{x}+\partial_{x}w\right)^{2}+\left(\psi_{y}+\partial_{y}w\right)^{2}\right],\nonumber 
\end{gather}
where $T$ is the kinetic energy density per unit of area and $U$
is the strain energy density per unit of area. Constant $\kappa$
that appears in expression (\ref{eq:mireco1c}) for the strain energy
$U$ is a \emph{shear correction coefficient} introduced for the Timoshenko
beam \citet[Sec. 5.2.1]{Magrab}. A typical value of shear correction
coefficient $\kappa$ is $\kappa=\frac{5}{6}$, \citet[Sec. 5.2.1]{Magrab}.
Constant $G$ is the \emph{shear modulus} and constant $D$ is the
\emph{flexural rigidity of the plate} defined as follows, \citet[Sec. 6.2.1, 7.1.2]{Magrab},
\citet[Sec. 5.1]{Langh}, \citet[ Sec. 4.4.5]{GerRix}: 
\begin{equation}
G=\frac{E}{2\left(1+\nu\right)},\quad D=\frac{Eh^{3}}{12\left(1-\nu^{2}\right)}.\label{eq:mireco1d}
\end{equation}

One recovers the Kirchhoff (classical) plate theory Lagrangian $L$
defined by equation (\ref{eq:LagLag1a}) from the Mindlin-Reissner
plate theory Lagrangian $L$ defined by equations (\ref{eq:mireco1b})
and (\ref{eq:mireco1c}) by (i) setting $\psi_{x}=-\partial_{x}w$,
$\psi_{y}=-\partial_{y}w$ in the strain energy density $U$ expression,
that is no shear strain contribution; (ii) removing terms involving
$\partial_{t}\psi_{x}$ and $\partial_{t}\psi_{y}$ from the kinetic
energy $T$ expression, that is no rotary motion contribution, \citet[Sec. 8.1.1]{Graff}.

The Euler-Lagrange equations corresponding Lagrangian $L$ defined
by equations (\ref{eq:mireco1b}) and (\ref{eq:mireco1c}) are, \citet[Sec. 7.1.3]{Magrab},
\citet[Sec. 14.9.2]{RaoVCS} : 
\begin{equation}
\frac{\rho h^{3}}{12}\partial_{t}^{2}\psi_{y}+\kappa hG\left(\psi_{y}+\partial_{y}w\right)-\frac{D}{2}\left[\left(1-\nu\right)\Delta\psi_{y}+\left(1+\nu\right)\partial_{y}\Phi\right]=0,\label{eq:mireco2a}
\end{equation}
\begin{equation}
\frac{\rho h^{3}}{12}\partial_{t}^{2}\psi_{x}+\kappa hG\left(\psi_{x}+\partial_{x}w\right)-\frac{D}{2}\left[\left(1-\nu\right)\Delta\psi_{x}+\left(1+\nu\right)\partial_{x}\Phi\right]=0,\label{eq:mireco2b}
\end{equation}
\begin{equation}
\rho h\partial_{t}^{2}w-\kappa hG\left(\Delta w+\Phi\right)=0,\quad\Phi=\partial_{x}\psi_{x}+\partial_{y}\psi_{y}.\label{eq:mireco2c}
\end{equation}
We will refer to the EL (\ref{eq:mireco2a})-(\ref{eq:mireco2c})
as Mindlin-Reissner plate equation or MR equations for short.

To construct a factorized form of the dispersion relations related
to the MR equations (\ref{eq:mireco2a})-(\ref{eq:mireco2c}) we would
like to embed Lagrangian $L$ into a family of Lagrangians $L_{b}$
where $b$ is real-valued parameter as follows: 
\begin{equation}
L_{b}=T-U_{b},\quad T=\frac{\rho h}{2}\left[\frac{h^{2}}{12}\left(\partial_{t}\psi_{x}\right)^{2}+\frac{h^{2}}{12}\left(\partial_{t}\psi_{y}\right)^{2}+\left(\partial_{t}w\right)^{2}\right],\label{eq:mireco3a}
\end{equation}
\begin{gather}
U_{b}=\frac{D}{2}\left[\left(\partial_{x}\psi_{x}\right)^{2}+\left(\partial_{y}\psi_{y}\right)^{2}+2\nu\left(\partial_{x}\psi_{x}\right)\left(\partial_{y}\psi_{y}\right)+\frac{1-\nu}{2}\left(\partial_{y}\psi_{x}+\partial_{x}\psi_{y}\right)^{2}\right]\label{eq:mireco3b}\\
+\frac{\kappa hG}{2}\left[\left(b\psi_{x}+\partial_{x}w\right)^{2}+\left(b\psi_{y}+\partial_{y}w\right)^{2}\right],\nonumber 
\end{gather}
Note that Lagrangian $L_{0}$ represents a system for which field
$w$ and fields $\psi_{x}$, $\psi_{y}$ don't interact and Lagrangian
$L_{1}$ is exactly Lagrangian $L$ for the Mindlin-Reissner theory
defined by equations (\ref{eq:mireco1b}) and (\ref{eq:mireco1c}).
These facts justifies naming $b$ a \emph{coupling parameter}. The
presence of coupling parameter $b$ in expressions for quantities
of interest is helpful in assessing the effect of interaction between
field $w$ and fields $\psi_{x}$, $\psi_{y}$ on those quantities.

The EL equations for Lagrangian $L_{b}$ are as follows: 
\begin{equation}
\frac{\rho h^{3}}{12}\partial_{t}^{2}\psi_{y}+\kappa hG\left(b\psi_{y}+\partial_{y}w\right)-\frac{D}{2}\left[\left(1-\nu\right)\Delta\psi_{y}+\left(1+\nu\right)\partial_{y}\Phi\right]=0,\label{eq:mireco4a}
\end{equation}
\begin{equation}
\frac{\rho h^{3}}{12}\partial_{t}^{2}\psi_{x}+\kappa hG\left(b\psi_{x}+\partial_{x}w\right)-\frac{D}{2}\left[\left(1-\nu\right)\Delta\psi_{x}+\left(1+\nu\right)\partial_{x}\Phi\right]=0,\label{eq:mireco4b}
\end{equation}
\begin{equation}
\rho h\partial_{t}^{2}w-\kappa hG\left(\Delta w+b\Phi\right)=0,\quad\Phi=\partial_{x}\psi_{x}+\partial_{y}\psi_{y}.\label{eq:mireco4c}
\end{equation}
Note that in the case of $b=1$ the EL equations (\ref{eq:mireco4a})-(\ref{eq:mireco4c})
are identical to the MR equations (\ref{eq:mireco2a})-(\ref{eq:mireco2c})
and from now on we refer to them as the Mindlin-Reissner equations.

To obtain the dispersion relations associated with the MR equations
(\ref{eq:mireco4a})-(\ref{eq:mireco4c}) we consider the system eigenmodes
represented as follows: 
\begin{gather}
w\left(x,y,t\right)=\hat{w}\left(k,\omega\right)\mathrm{e}^{-\mathrm{i}\left(\omega t-k_{x}x-k_{y}y\right)},\quad k=\left(k_{x},k_{y}\right),\label{eq:mireco5a}\\
\psi_{x}\left(x,y,t\right)=\hat{\psi_{x}}\left(k,\omega\right)\mathrm{e}^{-\mathrm{i}\left(\omega t-k_{x}x-k_{y}y\right)},\quad\psi_{y}\left(x,t\right)=\hat{\psi_{y}}\left(k,\omega\right)\mathrm{e}^{-\mathrm{i}\left(\omega t-k_{x}x-k_{y}y\right)},\nonumber 
\end{gather}
where $\omega$ and $k=k\left(\omega\right)$ are the frequency and
the wavenumber, respectively. The Fourier transformation (see Appendix
\ref{sec:four}) in time $t$ and space variables $x,y$ of the Euler-Lagrange
equations (\ref{eq:mireco4a})-(\ref{eq:mireco4c}) can be written
in the following matrix form: 
\begin{equation}
B_{b}X=0,\quad B_{b}\overset{\mathrm{def}}{=}\left[\begin{array}{rr}
\mathsf{A}_{b} & -\mathrm{i}b\kappa k_{y}hG\mathsf{k}\\
\mathrm{i}b\kappa hG\mathsf{k}^{\mathrm{T}} & h\left(\rho\omega^{2}-\kappa Gk^{2}\right)
\end{array}\right],\quad X\overset{\mathrm{def}}{=}\left[\begin{array}{r}
\hat{\Psi}\left(k,\omega\right)\\
\hat{w}\left(k,\omega\right)
\end{array}\right],\quad\hat{\Psi}\overset{\mathrm{def}}{=}\left[\begin{array}{r}
\hat{\psi_{y}}\\
\hat{\psi_{x}}
\end{array}\right],\label{eq:mireco5b}
\end{equation}
where $B_{b}$ is $3\times3$ is a Hermitian matrix, $\mathsf{A}$
is $2\times2$ matrix and $\mathsf{k}$ is $2\times1$ matrix (vector)
defined as follows: 
\begin{equation}
\mathsf{A}_{b}\overset{\mathrm{def}}{=}\left[\begin{array}{rr}
\frac{\rho h^{3}}{12}\omega^{2}-\kappa hGb^{2}-D\left(\frac{1-\nu}{2}k_{x}^{2}+k_{y}^{2}\right) & -\frac{D\left(1+\nu\right)}{2}k_{x}k_{y}\\
-\frac{D\left(1+\nu\right)}{2}k_{x}k_{y} & \frac{\rho h^{3}}{12}\omega^{2}-\kappa hGb^{2}-D\left(\frac{1-\nu}{2}k_{y}^{2}+k_{x}^{2}\right)
\end{array}\right],\quad\mathsf{k}\overset{\mathrm{def}}{=}\left[\begin{array}{r}
k_{y}\\
k_{x}
\end{array}\right].\label{eq:mireco5c}
\end{equation}

It turns out that matrix $B_{b}$ has a block diagonal form which
is as follows. Let us introduce the following orthonormal basis in
$\mathbb{R}^{3}$: 
\begin{equation}
\tau_{1}=\left[\begin{array}{r}
0\\
0\\
1
\end{array}\right],\quad\tau_{2}=\left[\begin{array}{r}
\frac{k_{x}}{k}\\
\frac{k_{y}}{k}\\
0
\end{array}\right],\quad\tau_{3}=\left[\begin{array}{r}
-\frac{k_{x}}{k}\\
\frac{k_{y}}{k}\\
0
\end{array}\right],\quad\left(\tau_{j},\tau_{m}\right)=\delta_{jm},\quad j,m=1\ldots3,\label{eq:mireco5ca}
\end{equation}
where $\left(\cdot,\cdot\right)$ is the scalar product in $\mathbb{C}^{3}$.
Then using vectors (\ref{eq:mireco5ca}) we define the following $3\times3$
matrix 
\begin{equation}
T_{k}\overset{\mathrm{def}}{=}\left[\tau_{1},\tau_{2},\tau_{3}\right]=\left[\begin{array}{rrr}
0 & \frac{k_{y}}{k} & -\frac{k_{x}}{k}\\
0 & \frac{k_{x}}{k} & \frac{k_{y}}{k}\\
1 & 0 & 0
\end{array}\right],\quad k=\sqrt{k_{x}^{2}+k_{y}^{2}}.\label{eq:mireco5cb}
\end{equation}
Note that vector $\tau_{2}$ corresponds to longitudinal (irrotational,
dilational) mode of oscillations, whereas vectors $\tau_{1}$ and$\tau_{3}$
represent transverse (equivoluminal, distortional) modes of oscillations.

It is straightforward to verify that for any $b$ Hermitian matrix
$B_{b}$ satisfies the following representation: 
\begin{equation}
B_{b}=T_{k}C_{b}T_{k}^{-1},\quad T_{k}^{-1}=T_{k}^{\mathrm{T}}=\left[\begin{array}{rrr}
0 & 0 & 1\\
\frac{k_{y}}{k} & \frac{k_{x}}{k} & 0\\
-\frac{k_{x}}{k} & \frac{k_{y}}{k} & 0
\end{array}\right],\quad k=\sqrt{k_{x}^{2}+k_{y}^{2}},\label{eq:mireco5d}
\end{equation}
where $C_{b}$ is the block-diagonal Hermitian matrix $3\times3$
defined by 
\begin{gather}
C_{b}\overset{\mathrm{def}}{=}\left[\begin{array}{rrr}
h\left(\rho\omega^{2}-\kappa Gk^{2}\right) & \mathrm{i}b\kappa hkG & 0\\
-\mathrm{i}b\kappa hkG & g\left(k,\omega\right) & 0\\
0 & 0 & f\left(k,\omega\right)
\end{array}\right],\label{eq:mireco5e}\\
g\left(k,\omega\right)\overset{\mathrm{def}}{=}\frac{\rho h^{3}}{12}\omega^{2}-Dk^{2}-b^{2}h\kappa G,\quad f\left(k,\omega\right)\overset{\mathrm{def}}{=}\frac{\rho h^{3}}{12}\omega^{2}-D\frac{\left(1-\nu\right)}{2}k^{2}-\kappa b^{2}hG.\label{eq:mireco5f}
\end{gather}
Note that equations (\ref{eq:mireco5ca}) and (\ref{eq:mireco5d})-(\ref{eq:mireco5f})
show that matrix $B_{b}$ can be block-diagonalized and that vector
$\tau_{3}$ is an eigenvector of matrix $B_{b}$, namely 
\begin{equation}
B_{b}\tau_{3}=f\left(k,\omega\right)\tau_{3}=\left(\frac{\rho h^{3}}{12}\omega^{2}-D\frac{\left(1-\nu\right)}{2}k^{2}-b^{2}\kappa hG\right)\tau_{3},\quad\tau_{3}=\left[\begin{array}{r}
-\frac{k_{x}}{k}\\
\frac{k_{y}}{k}\\
0
\end{array}\right].\label{eq:mireco5g}
\end{equation}
Note also that 
\begin{equation}
\left(B_{b}\tau_{j},\tau_{3}\right)=\left(\tau_{j},B_{b}\tau_{3}\right)=f\left(k,\omega\right)\left(\tau_{j},\tau_{3}\right)=0,\quad j=1,2,\label{eq:mireco5ga}
\end{equation}
implying that 
\begin{equation}
B_{b}\mathcal{T}\subseteq\mathcal{T},\quad\mathrm{\mathcal{T}=span}\,\left\{ \tau_{1},\tau_{2}\right\} ,\label{eq:mireco5gb}
\end{equation}
that is space $\mathcal{T}$is an invariant under action of matrix
$B_{b}$ subspace of $\mathbb{R}^{3}$. The latter is consistent with
equations (\ref{eq:mireco5d})-(\ref{eq:mireco5f}).

Equation of interest $B_{b}X=0$ in (\ref{eq:mireco5a}) has a nonzero
solution $X$ if and only if $\det\left\{ B_{b}\right\} =0$, and
the latter equation determines the dispersion relations associated
with the Mindlin-Reissner equations (\ref{eq:mireco4a})-(\ref{eq:mireco4c}).
A tedious but elementary analysis of equation $\det\left\{ B_{b}\right\} =\det\left\{ C_{b}\right\} =0$
and equations (\ref{eq:mireco5d})-(\ref{eq:mireco5f}) yield the
following \emph{factorized form of the dispersion relations for the
Mindlin-Reissner plate theory}: 
\begin{gather}
f\left(k,\omega\right)A\left(k,\omega\right)=0,\quad f\left(k,\omega\right)=\frac{\rho h^{3}}{12}\omega^{2}-D\frac{\left(1-\nu\right)}{2}k^{2}-b^{2}\kappa hG,\label{eq:mireco6a}\\
A\left(k,\omega\right)=\left(\frac{\rho h^{3}}{12}\omega^{2}-Dk^{2}\right)\left(\rho\omega^{2}-\kappa Gk^{2}\right)-b^{2}\kappa G\rho h\omega^{2}.\label{eq:mireco6b}
\end{gather}
Consequently, each pairs $\left(k,\omega\right)$ which is a solution
to the dispersion equations (\ref{eq:mireco6a}), (\ref{eq:mireco6b})
must be also a solution to at least one of the following equations:
\begin{equation}
f\left(k,\omega\right)=0\text{ or equivalently }\frac{\rho h^{3}}{12}\omega^{2}-D\frac{\left(1-\nu\right)}{2}k^{2}-b^{2}\kappa hG=0,\label{eq:mireco6c}
\end{equation}
\begin{equation}
A\left(k,\omega\right)=0\text{ or equivalently }\left(\frac{\rho h^{3}}{12}\omega^{2}-Dk^{2}\right)\left(\rho\omega^{2}-\kappa Gk^{2}\right)=b^{2}\kappa G\rho h\omega^{2}.\label{eq:mireco6d}
\end{equation}

It is instructive to consider the following alternative approach for
obtaining dispersion relations (\ref{eq:mireco6a})-(\ref{eq:mireco6b}).
Let us solve the first two equations of the system $B_{b}X=0$ for
fields $\hat{\psi_{y}}$ and $\hat{\psi_{x}}$ and obtain their representation
in terms of $\hat{w}$. If we then plug in the obtained expressions
into the third equation of the system $B_{b}X=0$ we find that the
resulting equation to be of the form $A\left(k,\omega\right)\hat{w}=0$,
where $A\left(k,\omega\right)$ is defined by equation (\ref{eq:mireco6b}).
The latter equation evidently has a nontrivial solution $\hat{w}$
if and only if $A_{w}\left(k,\omega\right)=0$. Repeating similar
developments for variables $\hat{\psi_{x}}$ and $\hat{w}$ and for
variables $\hat{\psi_{y}}$ and $\hat{w}$ we obtain equation $f\left(k,\omega\right)A\left(k,\omega\right)=0$
for the both cases. The described alternative approach for obtaining
the dispersion relations (\ref{eq:mireco6a})-(\ref{eq:mireco6b})
is consistent with somewhat different approach used by S. Rao, \citet[Sec. 14.9]{RaoVCS}.
Namely the author solves the MR equations (\ref{eq:mireco4a})-(\ref{eq:mireco4c})
for fields $\psi_{x}$, $\psi_{y}$ and $w$ for case of free vibrations
assuming the fields to be time-harmonic. Approach pursued by Rao in
\citet[Sec. 14.9.3]{RaoVCS} is based on a representations of fields
$\psi_{x}$, $\psi_{y}$ in terms of the Lame potentials that correspond
to the dilatation and shear components of motion of the plate, \citet[Sec. 7.3, 7.9]{ErinSuh}.

The dispersion equation (\ref{eq:mireco6d}), that is $A\left(k,\omega\right)=0$,
can be readily recast into the velocity dispersion relation, namely
\begin{equation}
\frac{1}{12}h^{2}k^{2}\left(1-\frac{c^{2}}{\kappa c_{\mathrm{T}}^{2}}\right)\left(\frac{c_{\mathrm{P}}^{2}}{c^{2}}-1\right)=b^{2},\quad c=\frac{\omega}{k},\label{eq:mireco7a}
\end{equation}
where $c_{\mathrm{L}}$ and $c_{\mathrm{T}}$ are respectively the
longitudinal and the transverse wave speeds for 3D homogeneous isotropic
elastic medium free of body forces defined as follows 
\begin{equation}
c_{\mathrm{L}}=\sqrt{\frac{\lambda+2G}{\rho}}=\sqrt{\frac{E\left(1-\nu\right)}{\rho\left(1+\nu\right)\left(1-2\nu\right)}},\quad c_{\mathrm{T}}=\sqrt{\frac{G}{\rho}}=\sqrt{\frac{E}{2\rho\left(1+\nu\right)}}.\label{eq:mireco7b}
\end{equation}
As to velocity $c_{\mathrm{P}}$ it is the velocity of the so-called
extensional waves in thin plates defined by the following formula,
\citet{Bish}, \citet[Sec. 2.7, 6.12.3]{Achen}, \citet[Sec. 7.3]{ErinSuh},
\citet[Sec. 4.3.2, 8.3.1]{Graff}: 
\begin{equation}
c_{\mathrm{P}}=\sqrt{\frac{E}{\rho\left(1-\nu^{2}\right)}}.\label{eq:mireco7c}
\end{equation}
Indeed, an analysis of the plane-stress problem for thin plates shows
that the passage from the plane-strain problem to the corresponding
plane-stress problem can be made by simply replacing the Lame constant
$\lambda$ with its modified value $\lambda^{\prime}$, namely, \citet{Bish},
\citet[Sec. 2.7, 6.12.3]{Achen}, \citet[Sec. 7.3, 7.9]{ErinSuh}:
\begin{equation}
\lambda^{\prime}=\frac{2G\lambda}{\lambda+2G}=\frac{\nu E}{1-\nu^{2}}.\label{eq:mireco8a}
\end{equation}
Using this modified value $\lambda^{\prime}$ of the Lame constant
in place of $\lambda$ and the expression for longitudinal wave speed
$c_{\mathrm{L}}=\sqrt{\frac{\lambda+2G}{\rho}}$ one obtains expression
(\ref{eq:mireco7c}) of the phase velocity $c_{\mathrm{P}}$, namely
\begin{equation}
c_{\mathrm{P}}=\sqrt{\frac{\lambda^{\prime}+2G}{\rho}}=\sqrt{\frac{E}{\rho\left(1-\nu^{2}\right)}}.\label{eq:mireco8b}
\end{equation}
Equations (\ref{eq:mireco7b}) and (\ref{eq:mireco7c}) readily imply
the following relations between velocities $c_{\mathrm{T}}$ and $c_{\mathrm{P}},$
\begin{equation}
\frac{c_{\mathrm{T}}^{2}}{c_{\mathrm{P}}^{2}}=\frac{1-\nu}{2},\label{eq:mireco7d}
\end{equation}

Note that dispersion relation (\ref{eq:mireco7a}) matches exactly
the velocity dispersion relation in \citet[Sec. 8.3.1]{Graff} for
the special case when $b=1$ corresponding to the Mindlin-Reissner
plate theory.

The dispersion relations (\ref{eq:mireco6c}) can be recast as the
\emph{velocity dispersion relation} as follows 
\begin{equation}
c=\frac{c_{\mathrm{T}}}{hk}\sqrt{h^{2}k^{2}+12\kappa b^{2}}=c_{\mathrm{T}}\sqrt{1+\frac{12\kappa b^{2}}{h^{2}k^{2}}},\quad c=\frac{\omega}{k}.\label{eq:mireco8c}
\end{equation}
Note that equation (\ref{eq:mireco8c}) implies the following asymptotic
formulas: 
\begin{equation}
c=c_{\mathrm{T}}\left[\frac{2b\sqrt{3\kappa}}{h}k^{-1}+\frac{h}{4b\sqrt{3\kappa}}k+O\left(k^{3}\right)\right],\quad c=\frac{\omega}{k},\quad k\rightarrow0,\label{eq:mireco8d}
\end{equation}
\begin{equation}
c=c_{\mathrm{T}}\left(1+\frac{6\kappa b^{2}}{h^{2}k^{2}}+O\left(k^{-4}\right)\right),\quad c=\frac{\omega}{k},\quad k\rightarrow\infty.\label{eq:mireco8e}
\end{equation}

It is instructive to consider a special case $b=0$ when fields $\psi_{x}$,
$\psi_{y}$ and $w$ are decoupled. In this case equations (\ref{eq:mireco5b}),
(\ref{eq:mireco5c}), (\ref{eq:mireco5e}) turn into 
\begin{equation}
B_{0}X=0,\quad B_{0}\overset{\mathrm{def}}{=}\left[\begin{array}{rr}
\mathsf{A}_{0} & 0\\
0 & h\left(\rho\omega^{2}-\kappa Gk^{2}\right)
\end{array}\right],\quad X\overset{\mathrm{def}}{=}\left[\begin{array}{r}
\hat{\Psi}\left(k,\omega\right)\\
\hat{w}\left(k,\omega\right)
\end{array}\right],\quad\hat{\Psi}\overset{\mathrm{def}}{=}\left[\begin{array}{r}
\hat{\psi_{y}}\\
\hat{\psi_{x}}
\end{array}\right],\label{eq:mireco5ba}
\end{equation}
where $\mathsf{A}_{0}$ is $2\times2$ matrix and $\mathsf{k}$ is
$2\times1$ matrix (vector) defined as follows: 
\begin{equation}
\mathsf{A}_{0}\overset{\mathrm{def}}{=}\left[\begin{array}{rr}
\frac{\rho h^{3}}{12}\omega^{2}-D\left(\frac{1-\nu}{2}k_{x}^{2}+k_{y}^{2}\right) & -\frac{D\left(1+\nu\right)}{2}k_{x}k_{y}\\
-\frac{D\left(1+\nu\right)}{2}k_{x}k_{y} & \frac{\rho h^{3}}{12}\omega^{2}-D\left(\frac{1-\nu}{2}k_{y}^{2}+k_{x}^{2}\right)
\end{array}\right],\quad\mathsf{k}\overset{\mathrm{def}}{=}\left[\begin{array}{r}
k_{y}\\
k_{x}
\end{array}\right],\label{eq:mireco5bb}
\end{equation}
\begin{gather}
C_{0}=\left[\begin{array}{rrr}
h\left(\rho\omega^{2}-\kappa Gk^{2}\right) & 0 & 0\\
0 & \frac{\rho h^{3}}{12}\omega^{2}-Dk^{2} & 0\\
0 & 0 & \frac{\rho h^{3}}{12}\omega^{2}-D\frac{\left(1-\nu\right)}{2}k^{2}
\end{array}\right],\label{eq:mireco5ea}
\end{gather}
The velocity dispersion equations (\ref{eq:mireco7a}) and (\ref{eq:mireco8c})
yield respectively the following non dispersive values of the characteristic
velocities when $b=0$: 
\begin{gather}
\mathrm{\mathcal{T}=span}\,\left\{ \tau_{1},\tau_{2}\right\} =\mathrm{span}\,\left\{ \left[\begin{array}{r}
0\\
0\\
1
\end{array}\right],\left[\begin{array}{r}
\frac{k_{x}}{k}\\
\frac{k_{y}}{k}\\
0
\end{array}\right]\right\} :\quad\sqrt{\kappa}c_{\mathrm{T}}=\sqrt{\frac{E\kappa}{2\rho\left(1+\nu\right)}},\quad c_{\mathrm{P}}=\sqrt{\frac{E}{\rho\left(1-\nu^{2}\right)}};\label{eq:mireco9a}\\
\tau_{3}=\left[\begin{array}{r}
-\frac{k_{x}}{k}\\
\frac{k_{y}}{k}\\
0
\end{array}\right]:\quad c_{\mathrm{T}}=\sqrt{\frac{E}{2\rho\left(1+\nu\right)}}.\label{eq:mireco9b}
\end{gather}
We remind that vector $\tau_{2}$ corresponds to longitudinal (irrotational,
dilational) mode of oscillations associated with velocity $c_{\mathrm{P}}$,
whereas vectors $\tau_{1}$ and$\tau_{3}$ represent transverse (equivoluminal,
distortional) modes of oscillations, having respectively velocities
$\sqrt{\kappa}c_{\mathrm{T}}$ and $c_{\mathrm{T}}$.

In the case of arbitrary $b$ vector $\tau_{3}$ represents oscillations
propagating at the speed satisfying the velocity dispersion equation
(\ref{eq:mireco8c}), whereas vectors in the invariant subspace $\mathrm{\mathcal{T}=span}\,\left\{ \tau_{1},\tau_{2}\right\} $
have characteristic velocities that satisfy the velocity dispersion
equation (\ref{eq:mireco7a}).

\subsubsection{Hybridization of modes}

\label{subsec:hybrid}

We analyze here the effect of coupling coefficient $b$ on dispersion
relations (\ref{eq:mireco6c}) and (\ref{eq:mireco6d}). We start
off with illustrating graphically the analytical developments of Section~\ref{subsec:midreis}
by plotting the dispersion relations (\ref{eq:mireco6c}) and (\ref{eq:mireco6d})
for three sets of data differing only in the value of the coupling
coefficient $b$: 
\begin{equation}
\rho=1,\quad h=1,\quad D=1,\quad\nu=\tfrac{1}{2},\quad\kappa=1,\quad G=1,\quad b=0,\quad b=0.1,\quad b=0.2.\label{eq:pLdisdat1a}
\end{equation}
The first set in (\ref{eq:pLdisdat1a}) has $b=0$ (no coupling),
while the other two sets have $b=0.1$ and $b=0.2$ respectively.
The dashed blue curve corresponds to $b=0$, the solid crimson curve
to $b=0.1$, and the solid dark green curve to $b=0.2$.

\paragraph{Dispersion relation $f(k,\omega)=0$.}

Figure~\ref{fig:plot_f} shows the dispersion relation $f(k,\omega)=0$,
namely 
\begin{equation}
\frac{\rho h^{3}}{12}\omega^{2}-D\frac{(1-\nu)}{2}k^{2}-b^{2}\kappa hG=0.\label{eq:pL_f}
\end{equation}
For $b=0$ this reduces to the straight-line pair $\omega=\pm\sqrt{6D(1-\nu)/(\rho h^{3})}\,k$
passing through the origin. For $b>0$ the coupling term $b^{2}\kappa hG$
shifts the branches away from the origin: the intercept at $k=0$
becomes $\omega_{0}=\pm b\sqrt{12\kappa G/(\rho h^{2})}$, so the
$f=0$ branches are lifted off the origin by the coupling.

\begin{figure}[h]
\begin{centering}
\includegraphics[scale=0.45]{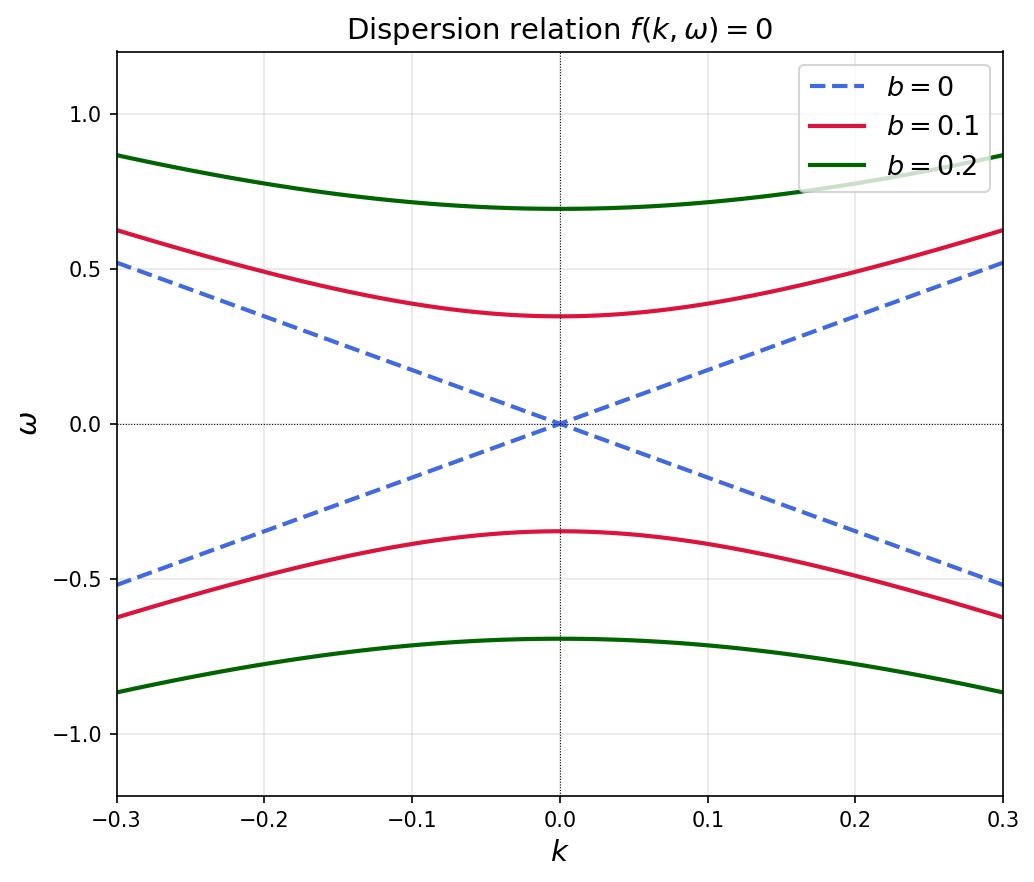} 
\par\end{centering}
\caption{ Dispersion relation $f(k,\omega)=0$ for data sets (\ref{eq:pLdisdat1a}).
Dashed blue: $b=0$; solid crimson: $b=0.1$; solid dark green: $b=0.2$.
The coupling lifts the solid branches off the origin.}
\label{fig:plot_f} 
\end{figure}

\paragraph{Dispersion relation $A(k,\omega)=0$.}

Figure~\ref{fig:plot_A} shows the dispersion relation $A(k,\omega)=0$,
namely 
\begin{equation}
\left(\frac{\rho h^{3}}{12}\omega^{2}-Dk^{2}\right)\left(\rho\omega^{2}-\kappa Gk^{2}\right)=b^{2}\kappa G\rho h\,\omega^{2},\label{eq:pL_A}
\end{equation}
together with a zoomed view near the origin. For $b=0$ the equation
factors into two pairs of straight lines through the origin: $\omega=\pm\sqrt{12D/(\rho h^{3})}\,k$
(from the first factor) and $\omega=\pm\sqrt{\kappa G/\rho}\,k$ (from
the second factor), giving four branches all pinned at $(k,\omega)=(0,0)$.
For $b>0$ the picture changes markedly: the two \emph{upper} branches
(larger $|\omega|$) are lifted off the origin, while the two \emph{lower}
branches (smaller $|\omega|$) appear to remain pinned. The zoomed
plot (right panel of Figure~\ref{fig:plot_A}) confirms that the
lower branches do indeed pass through the origin, approaching it parabolically
rather than linearly.

\begin{figure}[h]
\begin{centering}
\hspace{-0.5cm}\includegraphics[scale=0.45]{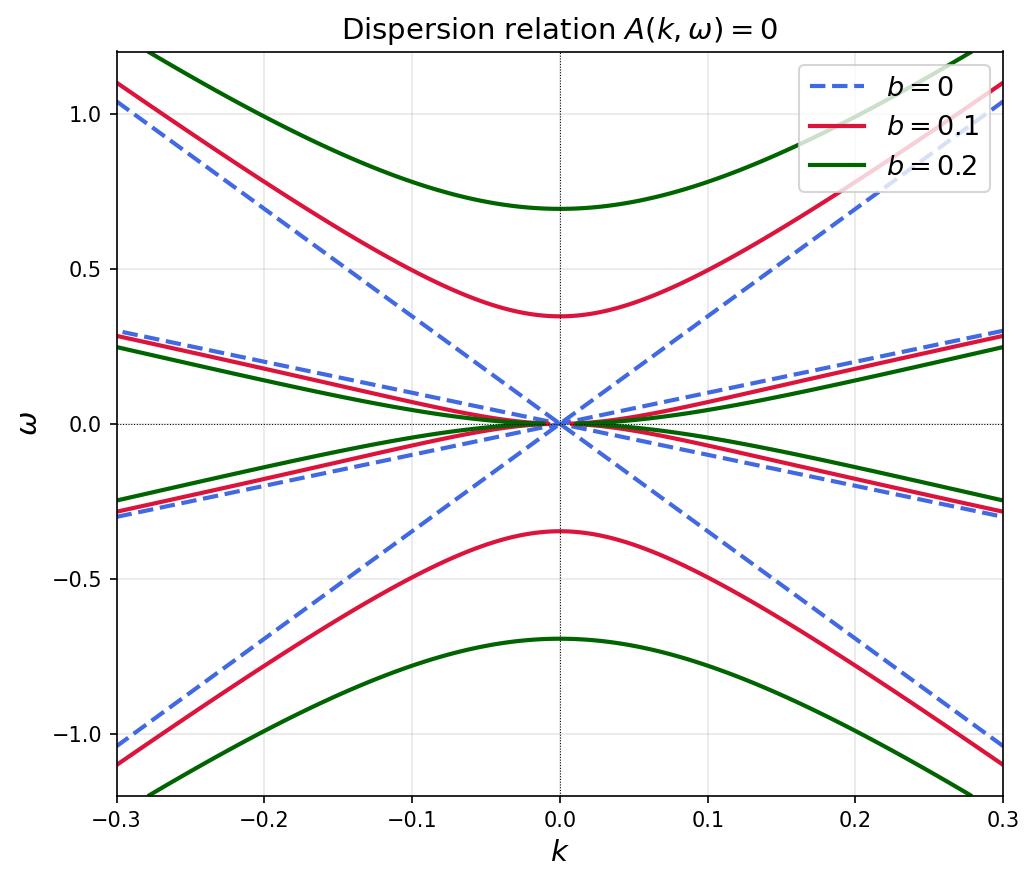}\hspace{0.5cm}\includegraphics[scale=0.45]{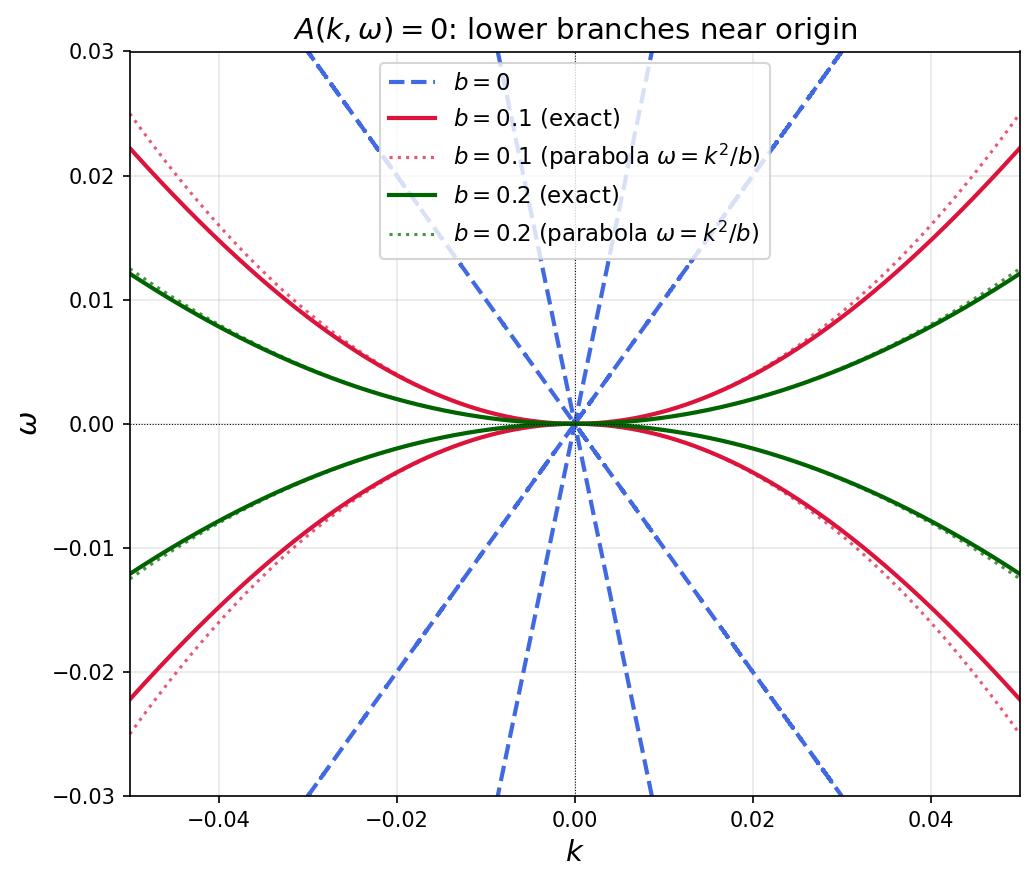} 
\par\end{centering}
(a)\hspace{7cm}(b)\caption{ Dispersion relation $A(k,\omega)=0$ for data set (\ref{eq:pLdisdat1a}).
Dashed blue: $b=0$; solid crimson: $b=0.1$; solid dark green: $b=0.2$.
(a)~Full view, $k\in(-0.3,0.3)$; (b)~zoomed view of lower branches
near origin, $k\in(-0.05,0.05)$, with dotted curves showing the parabolic
approximation $\omega=k^{2}/b$. The two upper branches are lifted
off the origin by coupling, while the two lower branches are pinned
at the origin with parabolic tangency.}
\label{fig:plot_A} 
\end{figure}

\paragraph{Asymptotic analysis: hybridization of modes.}

We now carry out an asymptotic analysis of $A(k,\omega)=0$ near the
origin to determine precisely which factors --- and hence which modes
--- govern each branch. Dividing equation~(\ref{eq:pL_A}) by $\omega^{4}$
and introducing $S=k^{2}/\omega^{2}$ we obtain the quadratic 
\begin{equation}
\kappa GD\cdot S^{2}-\left(\frac{\rho h^{3}\kappa G}{12}+\rho D\right)S+\frac{\rho^{2}h^{3}}{12}=\frac{b^{2}\kappa G\rho h}{\omega^{2}}.\label{eq:pL_Squad}
\end{equation}
Introducing the shorthand 
\begin{equation}
P=\frac{\rho h^{3}\kappa G}{12}+\rho D,\quad Q=\frac{\rho h^{3}\kappa G}{12}-\rho D,\quad R=2b\kappa G\sqrt{D\rho h},\label{eq:pL_PQR}
\end{equation}
the two solutions are 
\begin{equation}
S_{\pm}=\frac{P\pm\sqrt{Q^{2}+R^{2}/\omega^{2}}}{2\kappa GD}.\label{eq:pL_Spm}
\end{equation}
Note that $P$ and $Q$ involve parameters from \emph{both} factors:
$\rho D$ from the first and $\rho h^{3}\kappa G/12$ from the second,
while $R$ involves the coupling $b$ together with parameters from
both factors. Pulling out $1/|\omega|$ from the square root and expanding
$\sqrt{1+\varepsilon}=1+\frac{1}{2}\varepsilon-\frac{1}{8}\varepsilon^{2}+\cdots$
with $\varepsilon=Q^{2}\omega^{2}/R^{2}$, we obtain the Laurent series
for small $|\omega|$: 
\begin{equation}
S_{+}=\frac{R}{2\kappa GD}\cdot\frac{1}{|\omega|}+\frac{P}{2\kappa GD}+\frac{Q^{2}}{4\kappa GDR}\,|\omega|-\frac{Q^{4}}{16\kappa GDR^{3}}\,|\omega|^{3}+\cdots,\label{eq:pL_Splus}
\end{equation}
\begin{equation}
S_{-}=-\frac{R}{2\kappa GD}\cdot\frac{1}{|\omega|}+\frac{P}{2\kappa GD}-\frac{Q^{2}}{4\kappa GDR}\,|\omega|+\frac{Q^{4}}{16\kappa GDR^{3}}\,|\omega|^{3}+\cdots.\label{eq:pL_Sminus}
\end{equation}
Both series have been verified by direct substitution into the quadratic~(\ref{eq:pL_Squad}):
every coefficient from $|\omega|^{-2}$ through $|\omega|^{3}$ vanishes
identically.

The two series have opposite signs in the singular leading term, with
decisive physical consequences. Since $k^{2}=\omega^{2}S$, we analyze
each branch in turn:

\emph{$S_{+}$ branch.} The leading term $+R/(2\kappa GD|\omega|)>0$
dominates as $\omega\to0$, so $S_{+}>0$ for all small $\omega$
and $k^{2}=\omega^{2}S_{+}\sim R|\omega|/(2\kappa GD)\to0$: this
is the \emph{lower pinned branch}, approaching the origin parabolically.
Inverting $k^{2}=\omega^{2}S_{+}$ as a Puiseux series $\omega=c_{1}k^{2}+c_{2}k^{4}+c_{3}k^{6}+\cdots$
and solving order by order (verified by direct substitution into $A(k,\omega)=0$)
gives: 
\begin{gather}
c_{1}=\frac{\sqrt{D}}{b\sqrt{\rho h}},\quad c_{2}=-\frac{\sqrt{D}\,(12D+\kappa Gh^{3})}{24\,\kappa G\,b^{3}\,h^{3/2}\sqrt{\rho}},\quad c_{3}=\frac{\sqrt{D}\,(4D+\kappa Gh^{3})(36D+\kappa Gh^{3})}{384\,\kappa^{2}G^{2}\,b^{5}\,h^{5/2}\sqrt{\rho}},\label{eq:pL_c1}
\end{gather}
so that 
\begin{equation}
\omega=\pm\left(c_{1}k^{2}+c_{2}k^{4}+c_{3}k^{6}+\cdots\right).\label{eq:pL_omegaseries}
\end{equation}

\emph{$S_{-}$ branch.} The leading term $-R/(2\kappa GD|\omega|)<0$
dominates as $\omega\to0$, so $S_{-}<0$ for small $|\omega|$, giving
$k^{2}=\omega^{2}S_{-}<0$ --- no real $k$ exists. The branch is
absent near the origin and only becomes physical (real $k$) once
$|\omega|$ exceeds the threshold $\omega_{0}$ found by setting $k=0$
in~(\ref{eq:pL_A}) and dividing by $\omega^{2}\neq0$: 
\begin{equation}
\omega_{0}^{2}=\frac{12b^{2}\kappa G}{\rho h^{2}},\qquad\omega_{0}=\frac{2b\sqrt{3\kappa G}}{\sqrt{\rho}\,h}.\label{eq:pL_omega0}
\end{equation}
This is the \emph{upper lifted branch}. Expanding around $\omega_{0}$
as $\omega=\omega_{0}+d_{1}k^{2}+d_{2}k^{4}+\cdots$ and solving order
by order gives: 
\begin{gather}
d_{1}=\frac{\sqrt{3}\left(D+\frac{\kappa Gh^{3}}{12}\right)}{\sqrt{\kappa G}\,b\,h^{2}\sqrt{\rho}},\quad d_{2}=-\frac{\sqrt{3}\,(144D^{2}+72D\kappa Gh^{3}+\kappa^{2}G^{2}h^{6})}{576\,(\kappa G)^{3/2}\,b^{3}\,h^{3}\sqrt{\rho}},\label{eq:pL_d1}
\end{gather}
so that 
\begin{equation}
\omega=\pm\left(\omega_{0}+d_{1}k^{2}+d_{2}k^{4}+\cdots\right).\label{eq:pL_omegaupper}
\end{equation}

Several conclusions follow from the series (\ref{eq:pL_omegaseries})
and (\ref{eq:pL_omegaupper}).

\emph{Pinning confirmed.} The lower two branches~(\ref{eq:pL_omegaseries})
are pinned at the origin with parabolic tangency $\omega\sim c_{1}k^{2}$.
The upper two branches~(\ref{eq:pL_omegaupper}) are lifted off the
origin to $\pm\omega_{0}$, where $\omega_{0}\to0$ as $b\to0$: they
exist only because of the coupling.

\emph{Both modes contribute at every order.} In the lower branch~(\ref{eq:pL_omegaseries}),
the leading coefficient $c_{1}=\sqrt{D}/(b\sqrt{\rho h})$ involves
only $D$, $\rho$, $h$ from the \emph{first} factor, so the parabolic
curvature is set by the first mode alone. However, $c_{2}$ already
involves $\kappa G$ from the \emph{second} factor, and all higher
coefficients mix both. In the upper branch (\ref{eq:pL_omegaupper}),
the threshold $\omega_{0}$ involves $\kappa G$ (second factor) and
$\rho h$ (first factor), and every coefficient $d_{1},d_{2},\ldots$
mixes parameters from both factors.

\emph{Hybridization.} The coupling $b\neq0$ results in hybridization
of the two uncoupled modes in all four branches of $A(k,\omega)=0$:
the lower branches are governed to leading order by the first factor
but receive second-factor corrections at every higher order, while
the upper branches are a genuinely hybrid phenomenon whose very existence
requires the interaction of both modes.

\emph{Large-$|\omega|$ and large-$|k|$ asymptotics: recovery of
pure modes.} The quadratic~(\ref{eq:pL_Squad}) reveals an elegant
complementary picture in the opposite limit. As $|\omega|\to\infty$
the right-hand side $b^{2}\kappa G\rho h/\omega^{2}\to0$, which is
algebraically identical to setting $b=0$. The asymptotic equation
for $S$ is therefore simply 
\begin{equation}
\kappa GD\cdot S^{2}-P\cdot S+\frac{\rho^{2}h^{3}}{12}=0,\label{eq:pL_Sinf}
\end{equation}
with two roots 
\begin{equation}
S_{\pm}^{\infty}=\frac{P\pm|Q|}{2\kappa GD}=\begin{cases}
\dfrac{\rho}{\kappa G}, & (+)\\[6pt]
\dfrac{\rho h^{3}}{12D}, & (-)
\end{cases}\label{eq:pL_Sinf2}
\end{equation}
yielding the asymptotic slopes 
\begin{equation}
\omega\sim\pm\sqrt{\frac{\kappa G}{\rho}}\,k\quad\text{and}\quad\omega\sim\pm\sqrt{\frac{12D}{\rho h^{3}}}\,k,\quad|\omega|,|k|\to\infty.\label{eq:pL_asymp}
\end{equation}
These are precisely the slopes of the two uncoupled ($b=0$) straight-line
branches. Hence all four coupled branches are asymptotically straight
lines at large $|\omega|$ and $|k|$, with slopes entirely determined
by the individual factors --- the coupling term $b^{2}\kappa G\rho h/\omega^{2}$
decays as $1/\omega^{2}$ and becomes negligible. The hybridization
is therefore a \emph{low-frequency, small-wavenumber phenomenon}:
it is most pronounced near the origin and fades away as $|\omega|,|k|\to\infty$,
where each branch asymptotically recovers the identity of a single
pure mode.

\emph{Comparison with the cross-point model and growth of hybridization
with coupling.} The structural parallel between $A(k,\omega)=0$ and
the cross-point model~(\ref{eq:GGgamG2e}) is illuminating. In both
cases the factorized left-hand side is a product of the two individual
mode dispersion functions, and the right-hand side is the coupling
term. Setting the right-hand side to zero recovers the uncoupled straight-line
branches; any nonzero right-hand side forces every branch to carry
the imprint of both factors. The cross-point model is in fact the
local (linearized near the crossing) approximation to the general
story: it applies near $(0,0)$ in the Mindlin-Reissner case just
as it does near any cross-point $(\omega_{0},k_{0})$ in the general
factorized system.

The degree of hybridization grows monotonically with the coupling
and is directly readable from the plots. In Figure~\ref{fig:dis-crpo-new}
the coupled branches progressively depart from the dashed uncoupled
reference lines as $\gamma$ increases: the avoided-crossing gap widens
and the branches curve more strongly, mixing the two modes ever more
thoroughly. For large $|\kappa|$ and $|\delta|$ the coupled branches
visibly return to the reference lines, confirming the asymptotic recovery
of pure modes. The same tendency appears in Figure \ref{fig:plot_A}:
as $b$ increases from $0$ to $0.1$ to $0.2$, the upper branches
are lifted higher ($\omega_{0}\propto b$) and the parabolic lower
branches open more slowly ($\omega\sim k^{2}/b$, so the curvature
decreases with $b$) --- both are signatures of stronger hybridization
near the origin.

\emph{Contrasting directions of hybridization growth.} In the cross-point
model larger $\gamma$ always widens the avoided-crossing gap, while
in the Mindlin-Reissner model larger $b$ pushes the lower parabolic
branches \emph{closer} to the $k$-axis (smaller $\omega$ for fixed
$k$) --- a subtler but equally unambiguous signature of increased
mode mixing. In both cases the asymptotic straight-line behavior at
large $|\omega|$ and $|k|$ is independent of the coupling strength,
confirming that hybridization is confined to the neighborhood of the
cross-point.

\emph{Summary.} The factorized form $G_{1}G_{2}=\gamma G_{\mathrm{c}}$
makes mode hybridization not merely a qualitative statement but a
quantitatively precise one: the coupling parameter $b$ controls the
degree of mixing at every order of the asymptotic expansions, and
the deviation of each coupled branch from the uncoupled reference
curves provides a direct measure of hybridization that grows with
$b$ near the origin and vanishes asymptotically at large frequencies
and wavenumbers.

\subsubsection{Classical Kirchhoff's plate theory}

\label{subsubsec:KirchPT}

We concisely review here the classical Kirchhoff's plate theory which
is analogous of the Bernoulli-Euler beam theory. In the case of small
deflections its Lagrangian is, \citet[Sec. 5.1, 8.8]{Langh}, \citet[Sec. 6.2.2]{Magrab},
\citet[ Sec. 4.4]{GerRix} 
\begin{equation}
L=\frac{1}{2}\rho h\left(\partial_{t}w\right)^{2}+D\left(1-\nu\right)\left(\partial_{x}^{2}w\partial_{y}^{2}w-\left(\partial_{xy}^{2}w\right)^{2}\right)-\frac{1}{2}D\left(\partial_{x}^{2}w+\partial_{y}^{2}w\right)^{2},\label{eq:LagLag1a}
\end{equation}
where $w=w\left(x,y,t\right)$ is the plate deflection, $h$ is the
plate thickness, $\nu$ is \emph{Poisson's ratio},$\rho$ is the mass
density and $D$ is the \emph{flexural rigidity} defined by, \citet[Sec. 6.2.2]{Magrab},
\citet[Sec. 5.1]{Langh}, \citet[ Sec. 4.4.5]{GerRix} 
\begin{equation}
D=\frac{Eh^{3}}{12\left(1-\nu^{2}\right)}.\label{eq:LagLag1b}
\end{equation}

The Euler-Lagrange equation corresponding to Lagrangian $L$ defined
by equation (\ref{eq:LagLag1a}) is, \citet[Sec. 5.1, 8.8]{Langh},
\citet[Sec. 6.2.3]{Magrab}, \citet[ Sec. 4.4.9]{GerRix} 
\begin{gather}
\rho h\partial_{t}^{2}w+D\varDelta^{2}w=0,\quad\varDelta^{2}=\left(\partial_{x}^{2}+\partial_{y}^{2}\right)^{2}=\partial_{x}^{4}+2\partial_{x}^{2}\partial_{y}^{2}+\partial_{y}^{4}.\label{eq:LagLag1c}
\end{gather}
The fundamental differential equation (\ref{eq:LagLag1c}) in the
classical theory of vibration of plates was derived by Lagrange, \citet[Sec. 8.8]{Langh}.

Note that the second term of the Lagrangian $L$ in equation (\ref{eq:LagLag1a})
makes no contribution to the Euler-Lagrange equation (\ref{eq:LagLag1c}).

We briefly review here the theory of bending of plates following \citet[Chap. 1.1]{TimWoi}: 
\begin{quotation}
`` ...the simple problem of the bending of a long rectangular plate
that is subjected to a transverse load that does not vary along the
length of the plate. The deflected surface of a portion of such a
plate at a considerable distance from the ends can be assumed cylindrical,
with the axis of the cylinder parallel to the length of the plate.
We can therefore restrict ourselves to the investigation of the bending
of an elemental strip cut from the plate by two planes perpendicular
to the length of the plate and a unit distance (say 1 in.) apart.
The deflection of this strip is given by a differential equation which
is similar to the deflection equation of a bent beam. '' 
\end{quotation}
Suppose that the plate has uniform thickness $h$ and let $xy$ be
the middle plane of the plate before loading. Let the $y$-axis coincide
with one of the longitudinal edges of the plate and let the positive
direction of the $z$ axis be downward and the plate is bend downwards
under the load. Suppose that $w$ is the deflection of the plate in
the $z$ direction and $w$ is assumed to be small. Then the dispersion
relations that correspond to equations (\ref{eq:LagLag1c}) are, \citet[Sec. 4.2.3]{Graff}
\begin{equation}
\rho h\omega^{2}=Dk^{4},\quad k^{2}=k_{x}^{2}+k_{y}^{2}.\label{eq:LagLag1d}
\end{equation}
The Kirchhoff theory involves a single field $w$ and its dispersion
relation~(\ref{eq:LagLag1d}) has no two-subsystem factorized structure;
it serves here as the classical reference theory against which the
richer Mindlin-Reissner framework, with its coupling between the transverse
deflection $w$ and the rotational fields $\psi_{x}$, $\psi_{y}$,
is to be compared. We remark that asymptotic approaches to factorizing
plate dispersion relations, in a spirit related to the present work,
have been developed by Kaplunov and collaborators, see e.g.~\citet{KapNolRog},
\citet{KapNob}, \citet{ChebKapRog}, \citet{AlzKapPri}, where polynomial
approximations of the Rayleigh-Lamb and Mindlin plate dispersion relations
are derived that effectively isolate individual wave branches. The
present approach differs in that the factorization is achieved algebraically
through the Lagrangian coupling-parameter framework rather than by
asymptotic expansion in a small parameter.

\section{Cross-point model for factorized dispersion relations}

\label{sec:cross-po-model}

The cross-point model introduced by us in \citet{FigFDT1} is arguably
the simplest model illustrating the effect of coupling on the dispersion
relations of a system composed of two interacting subsystems. We provide
here a brief review of this model.

Let us assume that there are two initially non-interacting systems
with the dispersion relations defined by equations 
\begin{equation}
G_{1}\left(k,\omega\right)=0,\quad G_{2}\left(k,\omega\right)=0.\label{eq:GGgamG1a}
\end{equation}
Suppose then that the two systems are coupled and the dispersion relations
for this interacting system is of the following factorized form 
\begin{equation}
G_{1}\left(k,\omega\right)G_{2}\left(k,\omega\right)=\gamma G_{\mathrm{c}}\left(k,\omega\right),\label{eq:GGgamG1b}
\end{equation}
where $\gamma$ is the \emph{coupling coefficient} and we refer to
$G_{\mathrm{c}}\left(k,\omega\right)$ as the \emph{coupling function}.
We assume variables $k$ and $\omega$ be real-valued or complex-valued.

Suppose now that $\left(\omega_{0},k_{0}\right)$ is a \emph{``cross-point''}
of the graphs of functions $G_{1}$ and $G_{2}$, that a point satisfying
the two dispersion relations (\ref{eq:GGgamG1a}), namely 
\begin{equation}
G_{1}\left(\omega_{0},k_{0}\right)=0,\quad G_{2}\left(\omega_{0},k_{0}\right)=0.\label{eq:GGgamG1c}
\end{equation}
Suppose also coupling parameter $\gamma$ to be small, and consider
solutions to equation (\ref{eq:GGgamG1b}) in a small vicinity of
point $\left(\omega_{0},k_{0}\right)$, that is 
\begin{equation}
\left(k,\omega\right)=\left(k_{0}+\kappa,\omega_{0}+\delta\right),\quad\left|\delta\right|,\left|\kappa\right|\ll1,\label{eq:GGgamG1d}
\end{equation}
Assuming that equations (\ref{eq:GGgamG1c}) and (\ref{eq:GGgamG1d})
hold and that $|\delta|$, $|\kappa|$ are small, that is 
\begin{equation}
\left|\kappa\right|\ll1,\quad\left|\delta\right|\ll1,\label{eq:GGgamG1s}
\end{equation}
we arrive at the following \emph{principal approximation} to the dispersion
equation (\ref{eq:GGgamG1b}) 
\begin{equation}
\left(g_{1\omega}\delta+g_{1k}\kappa\right)\left(g_{2\omega}\delta+g_{2k}\kappa\right)=\gamma g_{\mathrm{c}},\label{eq:GGgamG2a}
\end{equation}
where the constants $g_{j\omega}$, $g_{jk}$ and $g_{\gamma}$ are
defined by 
\begin{equation}
g_{j\omega}=\left(\partial_{\omega}G_{j}\right)\left(\omega_{0},k_{0}\right),\quad g_{jk}=\left(\partial_{k}G_{j}\right)\left(\omega_{0},k_{0}\right),\quad j=1,2;\quad g_{\mathrm{c}}=G_{\mathrm{c}}\left(\omega_{0},k_{0}\right).\label{eq:GGgamG2b}
\end{equation}

For generic values of coefficients $g_{j\omega}$, $g_{jk}$ and $g_{\gamma}$
for which $g_{1\omega}\delta+g_{1k}\kappa$ and $g_{2\omega}\delta+g_{2k}\kappa$
are linearly independent, we can transform equations (\ref{eq:GGgamG2a})
into a simple special form by the following change of coordinates
\begin{equation}
g_{1\omega}\delta+g_{1k}\kappa=\delta^{\prime}+\kappa^{\prime},\quad g_{2\omega}\delta+g_{2k}\kappa=\delta^{\prime}-\kappa^{\prime}.\label{eq:GGgamG2c}
\end{equation}
Indeed equation (\ref{eq:GGgamG2a}) can be recast in terms of these
variables as 
\begin{equation}
\delta^{\prime2}-\kappa^{\prime2}=\gamma g_{\mathrm{c}}.\label{eq:GGgamG2d}
\end{equation}
Note now that the graph of equation (\ref{eq:GGgamG2d}) is a hyperbola
implying that the graph of original equation (\ref{eq:GGgamG2a})
is a linear transformation of the hyperbola associated with special
form (\ref{eq:GGgamG2d}).

\emph{In summary, we may conclude that generically the graph of the
dispersion relations of two interacting systems in a vicinity of the
relevant intersection point is a linear transformation of the hyperbola
if the coupling parameter $\gamma$ is small.}

In case when $g_{1\omega}\neq0$ and $g_{2\omega}\neq0$ we can divide
both sides of equation (\ref{eq:GGgamG2a}) by $g_{1\omega}g_{2\omega}$
obtaining the following equivalent equation 
\begin{equation}
\left(\delta+g_{1}\kappa\right)\left(\delta+g_{2}\kappa\right)=\gamma g_{\gamma},\quad g_{1}\stackrel{\mathrm{def}}{=}\frac{g_{1k}}{g_{1\omega}},\quad g_{2}\stackrel{\mathrm{def}}{=}\frac{g_{2k}}{g_{2\omega}},\quad g_{\gamma}\stackrel{\mathrm{def}}{=}\frac{g_{\mathrm{c}}}{g_{1\omega}g_{2\omega}}.\label{eq:GGgamG2e}
\end{equation}
We will refer to dispersion equation (\ref{eq:GGgamG2e}) as the \emph{cross-point
principle model dispersion relations}, and Figure~\ref{fig:dis-crpo-new}
shows the plots of these relations for $g_{1}=1$, $g_{2}=10$, $\gamma=0.4,2,4$,
and both signs of $g_{\gamma}$. In the figure: solid curves show
the dispersion curves for the indicated values of $\gamma$ (royal
blue: $\gamma=0.4$, crimson: $\gamma=2$, dark green: $\gamma=4$);
dashed blue straight lines show the uncoupled ($\gamma=0$) reference
curves. 
\begin{figure}[h]
\begin{centering}
\includegraphics[width=0.95\textwidth]{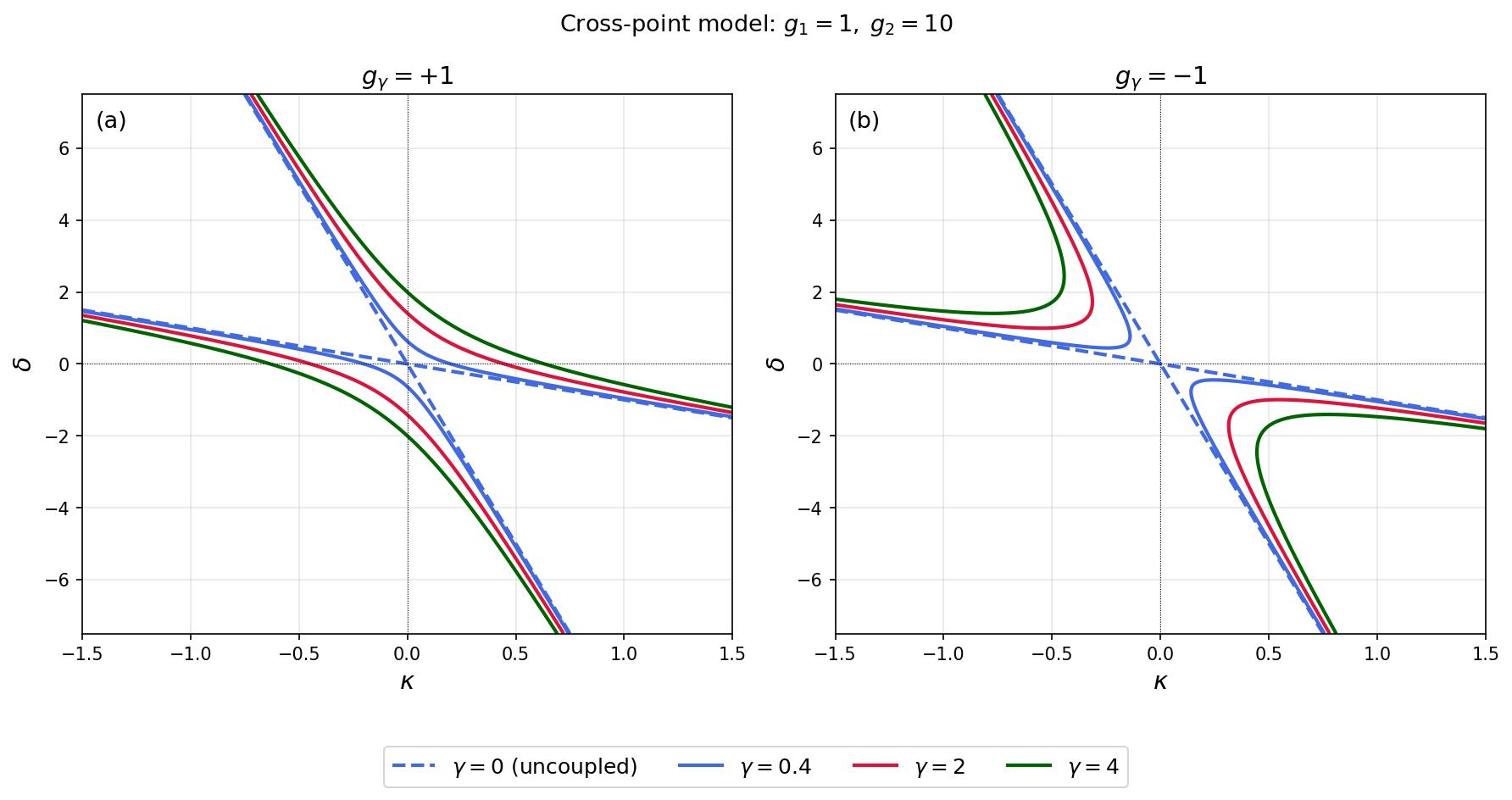} 
\par\end{centering}
\caption{ Cross-point principle model dispersion relations~(\ref{eq:GGgamG2e})
for $g_{1}=1$, $g_{2}=10$, $\gamma=0.4,2,4$: (a)~$g_{\gamma}=+1>0$;
(b)~$g_{\gamma}=-1<0$. Dashed blue: uncoupled reference lines ($\gamma=0$);
solid royal blue, crimson, dark green: coupled curves for $\gamma=0.4,2,4$
respectively. Curves farther from the dashed reference lines correspond
to larger values of $\gamma$.}
\label{fig:dis-crpo-new} 
\end{figure}

Note that coupled mode theory is yet another example that yields frequency
dependence on a parameter (detuning frequency) with graphical representation
\citet[Fig. 1]{HausHua} similar to Figure~\ref{fig:dis-crpo-new}.

The asymptotic behavior of the cross-point model is also worth noting.
For large $|\kappa|$ and $|\delta|$ the right-hand side $\gamma g_{\gamma}$
of~(\ref{eq:GGgamG2e}) becomes negligible compared to the left-hand
side, and the coupled branches asymptotically approach the uncoupled
reference lines $\delta=-g_{1}\kappa$ and $\delta=-g_{2}\kappa$.
More precisely, the deviation of each branch from the nearer reference
line decays as $O(1/\kappa)$ for large $|\kappa|$, since from~(\ref{eq:GGgamG2e})
the deviation $\Delta\delta$ satisfies $\Delta\delta\sim\gamma g_{\gamma}/((g_{2}-g_{1})\kappa)$.
This is the cross-point analog of the large-$|\omega|$ recovery of
pure modes established for $A(k,\omega)=0$ in Section~\ref{subsec:hybrid}
(equations~(\ref{eq:pL_Sinf})--(\ref{eq:pL_asymp})), where the
coupling term decays as $b^{2}/\omega^{2}$: in both models the coupling
term on the right-hand side of the factorized equation becomes relatively
small far from the cross-point, so the hybridization is spatially
concentrated near the crossing and the branches recover their individual
mode character at large $|\kappa|$, $|\delta|$ (or equivalently
large $|k|$, $|\omega|$). This behavior is clearly visible in Figure~\ref{fig:dis-crpo-new}:
all coupled branches, regardless of the value of $\gamma$, asymptotically
converge to the two dashed reference lines.

\subsection{Lagrangian framework for the cross-point model}

\label{subsec:Lag-crpo}

Let us consider the following general form of the dispersion function
and the corresponding dispersion relations 
\begin{equation}
G\left(k,\omega\right)\stackrel{\mathrm{def}}{=}A\omega^{2}-2\omega kB-Ck^{2}-D,\quad G\left(k,\omega\right)=0.\label{eq:GABC1a}
\end{equation}
The choice of signs before coefficients in equations (\ref{eq:GABC1a})
is motivated by its applications to the GTL, the e-beam and other
physical systems.

It is natural and important to ask if the dispersion relations (\ref{eq:GABC1a})
can be associated with a ``real physical system'', that is with
the Euler-Lagrange equations of a Lagrangian. The answer to this question
is positive, and an expression for such a Lagrangian $\mathcal{L}_{G}$
is as follows: 
\begin{equation}
\mathcal{L}_{G}\left(\partial_{t}Q,\partial_{z}Q,Q\right)\stackrel{\mathrm{def}}{=}\frac{Q^{2}}{2}G\left(\frac{\partial_{z}Q}{Q},\frac{\partial_{t}Q}{Q}\right)\equiv\frac{1}{2}\left[A\left(\partial_{t}Q\right)^{2}+2B\partial_{t}Q\partial_{z}Q-C\left(\partial_{z}Q\right)^{2}-DQ^{2}\right],\label{eq:GABC1b}
\end{equation}
where $Q=Q\left(z,t\right)$. Indeed, the EL equations for Lagrangian
$\mathcal{L}_{G}$ defined by equations (\ref{eq:GABC1b}) are 
\begin{equation}
\left[A\partial_{t}^{2}+2B\partial_{t}\partial_{z}-C\partial_{z}^{2}+D\right]Q=0.\label{eq:GABC1c}
\end{equation}
To find the dispersion relations associated with the EL equation (\ref{eq:GABC1c})
we proceed in the standard fashion and consider the system eigenmodes
of the form 
\begin{equation}
Q\left(z,t\right)=\hat{Q}\left(k,\omega\right)\mathrm{e}^{-\mathrm{i}\left(\omega t-kz\right)}.\label{eq:GABC1d}
\end{equation}
Plugging in expression (\ref{eq:GABC1c}) for $Q\left(z,t\right)$
in the EL equation (\ref{eq:GABC1c}) after elementary evaluations
we obtain 
\begin{equation}
\mathrm{e}^{-\mathrm{i}\left(\omega t-kz\right)}\hat{Q}\left(k,\omega\right)\left[-A\omega^{2}+2\omega kB+Ck^{2}+D\right]=0.\label{eq:GABC1e}
\end{equation}
Assuming naturally that $\hat{Q}\left(k,\omega\right)$ being an amplitude
of an eigenmode is not zero we recover from equation (\ref{eq:GABC1e})
the following dispersion relation associated with the EL equation
(\ref{eq:GABC1c}) 
\[
-A\omega^{2}+2\omega kB+Ck^{2}+D=0,
\]
which is evidently equivalent to the original dispersion relation
(\ref{eq:GABC1a}). Hence indeed the Lagrangian $\mathcal{L}_{G}$
defined by equation (\ref{eq:GABC1b}) yields indeed the EL equation
having the desired dispersion relation (\ref{eq:GABC1a}).

Motivated by the cross-point dispersion relations (\ref{eq:GGgamG2e})
we introduce cross-point dispersion relations 
\begin{equation}
G_{\mathrm{crp}}\left(k,\omega\right)\stackrel{\mathrm{def}}{=}\left(\omega+g_{1}k\right)\left(\omega+g_{2}k\right)-\gamma g_{\gamma},\quad G_{\mathrm{crp}}\left(k,\omega\right)=0.\label{eq:GABC2a}
\end{equation}
Then according to formula (\ref{eq:GABC1b}) the corresponding to
dispersion relations (\ref{eq:GABC2a}) Lagrangian $\mathcal{L}_{\mathrm{crp}}$
is of the form 
\begin{equation}
\mathcal{L}_{\mathrm{crp}}=\frac{1}{2}\left[\left(\partial_{t}Q+g_{1}\partial_{z}Q\right)\left(\partial_{t}Q+g_{2}\partial_{z}Q\right)-\gamma g_{\gamma}Q^{2}\right].\label{eq:GABC2b}
\end{equation}
The expression (\ref{eq:GABC2b}) can be also readily obtained from
the last expression of relations (\ref{eq:GABC1b}) by setting up
there the following values of coefficients: 
\begin{equation}
A=1,\quad B=\frac{g_{1}+g_{2}}{2},\quad C=-g_{1}g_{2},\quad D=\gamma g_{\gamma}.\label{eq:GABC2c}
\end{equation}

\subsection{Mechanical analog of the cross-point model}

\label{subsec:mech-crpo}

The cross-point dispersion relation~(\ref{eq:GGgamG2e}) describes
the behavior of two coupled continuum subsystems near a crossing point
in the $\left(\omega,k\right)$ plane. We present here a finite-dimensional
mechanical analog in which the wavenumber $k$ is replaced by a scalar
parameter $p$, yielding a system whose eigenfrequencies exhibit the
same factorized structure and crossing behavior as (\ref{eq:GGgamG2e}).
The construction is based on the coupled-oscillator framework of \citet[Sec. 6.1]{Likh},
modified to introduce a $p$-dependent Lagrangian.

Consider two harmonic oscillators with masses $m_{j}$ and $b$-dependent
spring constants 
\begin{equation}
\kappa_{j}\left(b\right)=\kappa_{j}-b\kappa,\quad j=1,2,\label{eq:kappab}
\end{equation}
where $\kappa_{j}$ are the bare spring constants, $\kappa>0$ is
a coupling spring constant, and $b\geq0$ is a dimensionless coupling
amplitude. The uncoupled Lagrangian is $L=L_{1}+L_{2}$ with 
\begin{equation}
L_{j}=\frac{m_{j}}{2}\dot{x}_{j}^{2}-\frac{\kappa_{j}\!\left(b\right)+p\alpha_{j}\kappa_{j}}{2}x_{j}^{2},\quad j=1,2,\label{eq:LagMech}
\end{equation}
where $p$ is a real parameter with $\left|p\right|\leq1/5$ and $\alpha_{j}$
are fixed dimensionless coefficients. The full Lagrangian includes
the physically meaningful relative-displacement coupling 
\begin{equation}
L=L_{1}+L_{2}+L_{\mathrm{int}},\qquad L_{\mathrm{int}}=-\frac{b\kappa}{2}\left(x_{1}-x_{2}\right)^{2}.\label{eq:LagMechFull}
\end{equation}
Expanding $L_{\mathrm{int}}$ and combining with~(\ref{eq:LagMech}),
the effective diagonal potential for oscillator $j$ is 
\begin{equation}
\frac{\kappa_{j}\!\left(b\right)+p\alpha_{j}\kappa_{j}+b\kappa}{2}x_{j}^{2}=\frac{\left(1+p\alpha_{j}\right)\kappa_{j}}{2}x_{j}^{2},\label{eq:cancel}
\end{equation}
where the $b$-dependent terms cancel exactly. The full Lagrangian
therefore reduces to 
\begin{equation}
L=\sum_{j=1}^{2}\left[\frac{m_{j}}{2}\dot{x}_{j}^{2}-\frac{\left(1+p\alpha_{j}\right)\kappa_{j}}{2}x_{j}^{2}\right]+b\kappa x_{1}x_{2},\label{eq:LagEff}
\end{equation}
which is equivalent to a system with bare spring constants $\left(1+p\alpha_{j}\right)\kappa_{j}$
and a purely off-diagonal coupling $b\kappa x_{1}x_{2}$. The Euler--Lagrange
equations of (\ref{eq:LagEff}) are 
\begin{align}
m_{1}\ddot{x}_{1}+\left(1+p\alpha_{1}\right)\kappa_{1}x_{1} & =b\kappa x_{2},\label{eq:eomMech1}\\
m_{2}\ddot{x}_{2}+\left(1+p\alpha_{2}\right)\kappa_{2}x_{2} & =b\kappa x_{1}.\label{eq:eomMech2}
\end{align}
Defining the $p$-dependent partial frequencies 
\begin{equation}
\widetilde{\Omega}_{j}^{2}\left(p\right)=\frac{\left(1+p\alpha_{j}\right)\kappa_{j}}{m_{j}},\quad j=1,2,\label{eq:partfreqMech}
\end{equation}
and seeking solutions $x_{j}=c_{j}e^{i\omega t}$, the characteristic
determinant of~(\ref{eq:eomMech1})--(\ref{eq:eomMech2}) yields
the \emph{factorized characteristic equation} 
\begin{equation}
\left(\omega^{2}-\widetilde{\Omega}_{1}^{2}\left(p\right)\right)\left(\omega^{2}-\widetilde{\Omega}_{2}^{2}\left(p\right)\right)=\frac{b^{2}\kappa^{2}}{m_{1}m_{2}},\label{eq:factMech}
\end{equation}
the mechanical analog of the cross-point dispersion relation~(\ref{eq:GGgamG2e}),
with the wavenumber $k$ replaced by the parameter $p$ and the coupling
coefficient $\gamma g_{\gamma}$ replaced by $b^{2}\kappa^{2}/\!\left(m_{1}m_{2}\right)$.
The two real eigenfrequency branches are 
\begin{equation}
\omega_{\pm}^{2}\!\left(p,b\right)=\frac{\widetilde{\Omega}_{1}^{2}+\widetilde{\Omega}_{2}^{2}}{2}\pm\sqrt{\frac{\left(\widetilde{\Omega}_{1}^{2}-\widetilde{\Omega}_{2}^{2}\right)^{2}}{4}+\frac{b^{2}\kappa^{2}}{m_{1}m_{2}}}.\label{eq:omegapmMech}
\end{equation}
Note that for $p=0$ equation (\ref{eq:omegapmMech}) in view of (\ref{eq:partfreqMech})
turns into 
\begin{equation}
\omega_{\pm}^{2}\!\left(0,b\right)=\frac{\Omega_{1}^{2}+\Omega_{2}^{2}}{2}\pm\sqrt{\frac{\left(\Omega_{1}^{2}-\Omega_{2}^{2}\right)^{2}}{4}+\frac{b^{2}\kappa^{2}}{m_{1}m_{2}}},\quad\Omega_{j}^{2}=\widetilde{\Omega}_{1}^{2}\left(0\right)=\frac{\kappa_{j}}{m_{j}},\quad j=1,2.\label{eq:omegapmMech1}
\end{equation}
Equation (\ref{eq:omegapmMech1}) in turn readily implies 
\begin{equation}
\omega_{\pm}^{2}\!\left(0,0\right)=\frac{\left(\Omega_{1}^{2}+\Omega_{2}^{2}\right)\pm\left|\Omega_{1}^{2}-\Omega_{2}^{2}\right|}{2},\quad\Omega_{j}^{2}=\frac{\kappa_{j}}{m_{j}},\quad j=1,2.\label{eq:omegapmMech2}
\end{equation}

Since in view of (\ref{eq:partfreqMech}) $\widetilde{\Omega}_{j}^{2}$
is independent of $b$, the partial-frequency crossing condition $\widetilde{\Omega}_{1}^{2}\left(p^{*}\right)=\widetilde{\Omega}_{2}^{2}\left(p^{*}\right)$
is likewise $b$-independent, and yields 
\begin{equation}
p^{*}=\frac{\kappa_{2}/m_{2}-\kappa_{1}/m_{1}}{\alpha_{1}\kappa_{1}/m_{1}-\alpha_{2}\kappa_{2}/m_{2}}.\label{eq:pstar}
\end{equation}
At $p=p^{*}$ both branches split symmetrically about the common value
$\Omega^{*}=\widetilde{\Omega}_{j}\left(p^{*}\right)$: 
\begin{equation}
\omega_{\pm}^{2}\left(p^{*},b\right)=\Omega^{*2}\pm\frac{b\kappa}{\sqrt{m_{1}m_{2}}},\label{eq:splitmech}
\end{equation}
so that neither branch is pinned to $\Omega^{*}$ for $b>0$. This
symmetric avoided crossing is the direct mechanical counterpart of
the hyperbolic geometry of the cross-point dispersion relation~(\ref{eq:GGgamG2d}).

To illustrate these results numerically we set 
\begin{equation}
m_{1}=m_{2}=1,\quad\kappa_{1}=1,\quad\kappa_{2}=\tfrac{6}{5},\quad\kappa=1,\quad\alpha_{1}=1,\quad\alpha_{2}=-1,\label{eq:dataMech}
\end{equation}
giving $p^{*}=1/11$ and $\Omega^{*}=\sqrt{12/11}\approx1.044$. Figure
\ref{fig:eigenfreq-mech} shows $\omega_{\pm}\left(p\right)$ over
the range $p\in\left(-1/20,\,23/100\right)$, chosen so that $p^{*}$
lies near the center, for four values of the coupling parameter $b=0,\,0.2,\,0.4,\,0.6$.
For $b=0$ the two branches cross at $p^{*}$; for $b>0$ the crossing
is replaced by an avoided crossing whose gap $\omega_{+}-\omega_{-}$
at $p^{*}$ grows with $b$, in precise analogy with the cross-point
dispersion relation~(\ref{eq:GGgamG2e}). 
\begin{figure}[h]
\centering
\includegraphics[scale=0.5]{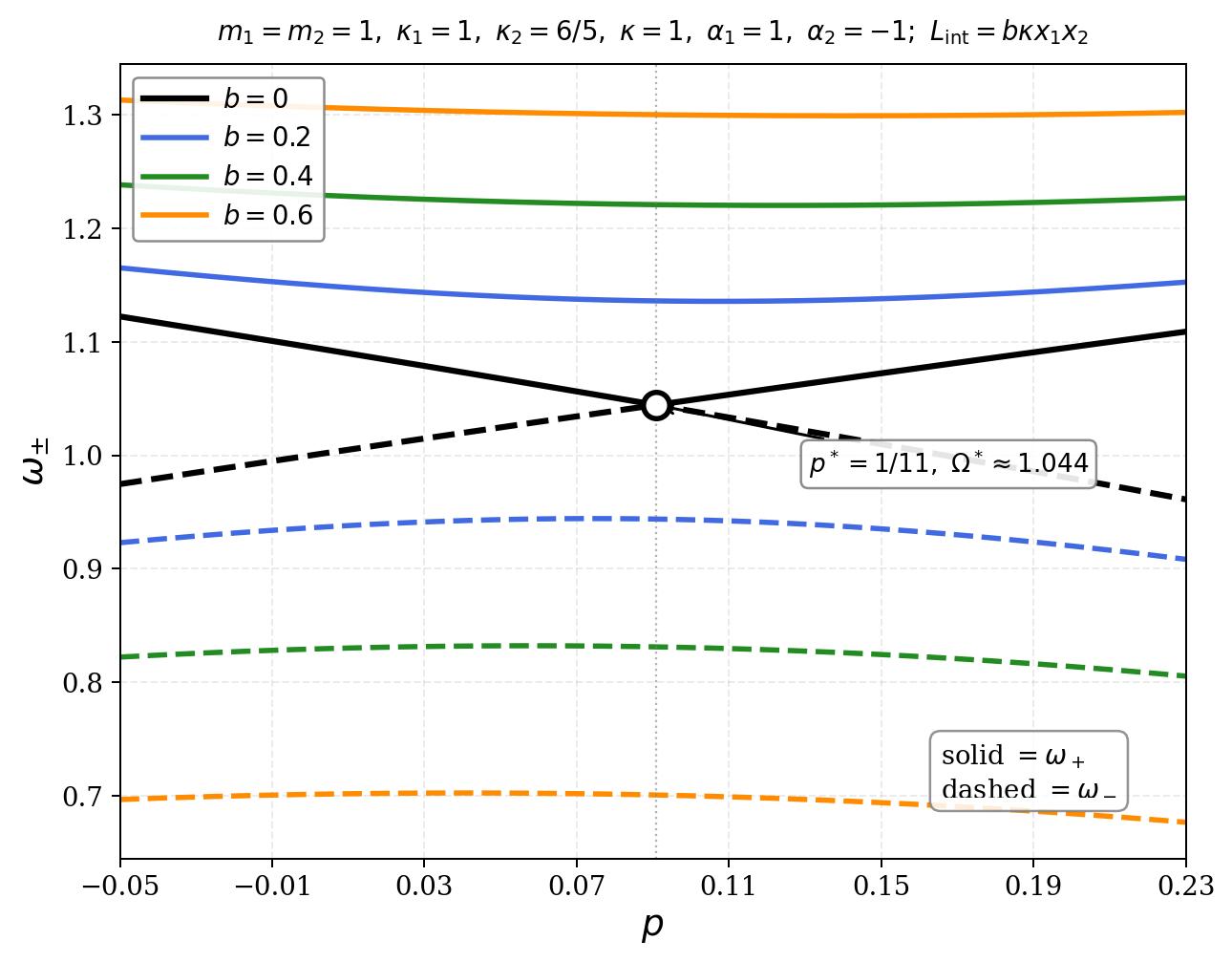} \caption{ Eigenfrequencies $\omega_{\pm}\left(p\right)$ for the parameter
set (\ref{eq:dataMech}) and $b=0$ (black), $0.2$ (blue), $0.4$
(green), $0.6$ (orange). Solid lines: upper branch $\omega_{+}$;
dashed lines: lower branch $\omega_{-}$. The open circle marks the
bare crossing point at $p^{*}=1/11$, $\omega=\Omega^{*}\approx1.044$.
The avoided crossing for $b>0$ is the mechanical analog of the cross-point
dispersion relation (\ref{eq:GGgamG2e}).}
\label{fig:eigenfreq-mech} 
\end{figure}

\textbf{\vspace{0.2cm}
 }

\textbf{ACKNOWLEDGMENT:} This research was supported by AFOSR MURI
Grant FA9550-20-1-0409 administered through the University of New
Mexico.

\textbf{\vspace{0.2cm}
 }

\section{Appendix}

\label{sec:append}

\subsection{Fourier transform}

\label{sec:four}

There are several common conventions for the Fourier transform, differing
in signs and constants. Our preferred form of the \emph{Fourier transform}
$\widehat{f}=f^{\land}$ of $f$ and the \emph{inverse Fourier transform}
$f^{\land}$ of $f$ follows to \citet[Sec. 1.1.7]{AdamHed}, \citet[Sec. 20.2]{ArfWeb},
\citet[Notations]{DauLio1}, \citet[Sec. 7.2, 7.5]{Foll}, \citet[Sec. 25]{TreB}:
\begin{gather}
\widehat{f}\left(k\right)\stackrel{\mathrm{def}}{=}\int_{-\infty}^{\infty}f\left(z\right)e^{-\mathrm{i}kz}\,dz,\quad f\left(z\right)=\left[\widehat{f}\left(k\right)\right]^{\lor}=\frac{1}{2\pi}\int_{-\infty}^{\infty}\widehat{f}\left(k\right)\mathrm{e}^{\mathrm{i}kz}\,\mathrm{d}k\label{eq:fouri1a}
\end{gather}
\begin{equation}
\widehat{f}\left(\omega\right)\stackrel{\mathrm{def}}{=}\int_{-\infty}^{\infty}f\left(t\right)e^{\mathrm{i}\omega t}\,\mathrm{d}t,\quad f\left(t\right)=\frac{1}{2\pi}\int_{-\infty}^{\infty}\widehat{f}\left(\omega\right)\mathrm{e}^{-\mathrm{i}\omega t}\,\mathrm{d}\omega,\label{eq:fouri1b}
\end{equation}
\begin{gather}
\widehat{f}\left(k,\omega\right)\stackrel{\mathrm{def}}{=}\int_{-\infty}^{\infty}f\left(z,t\right)e^{\mathrm{i}\left(\omega t-kz\right)}\,dz\mathrm{d}t,\label{eq:fouri1c}\\
f\left(z,t\right)=\left[\widehat{f}\left(k,\omega\right)\right]^{\lor}=\frac{1}{\left(2\pi\right)^{2}}\int_{-\infty}^{\infty}\widehat{f}\left(k,\omega\right)\mathrm{e}^{-\mathrm{i}\left(\omega t-kz\right)}\,\mathrm{d}k\mathrm{d}\omega.\nonumber 
\end{gather}
\emph{Note the difference of the choice of the sign for time $t$
and spatial variable $z$ in the above formula. It is motivated by
the desire to have ``wave'' form for exponential $\mathrm{e}^{-\mathrm{i}\left(\omega t-kz\right)}$
when both variables $t$ and $z$ are present}.

For multi-dimensional space variable $x\in\mathbb{R}^{n}$ the Fourier
transform $\widehat{f}$ of $f$ and the inverse Fourier transform
$f^{\land}$ of $f$ are defined by, \citet[Sec. 1.1.7]{AdamHed},
\citet[Notations]{DauLio1}, \citet[Sec. 7.5]{Foll}: 
\begin{equation}
\widehat{f}\left(k\right)\stackrel{\mathrm{def}}{=}\int_{\mathbb{R}^{n}}\widehat{f}\left(x\right)\mathrm{e}^{-\mathrm{i}k\cdot x}\,\mathrm{d}x,\quad f\left(x\right)=\left[\widehat{f}\left(k\right)\right]^{\lor}=\frac{1}{\left(2\pi\right)^{n}}\int_{\mathbb{R}^{n}}\widehat{f}\left(k\right)\mathrm{e}^{\mathrm{i}k\cdot x}\,\mathrm{d}k,\quad k,x\in\mathbb{R}^{n},\label{eq:fourier1d}
\end{equation}
which is consistent with equations (\ref{eq:fouri1a}). Then the Plancherel-Parseval
formula reads, \citet[Sec. 4.3.1]{Evans}, \citet[Sec. 7.5]{Foll},
\citet[Sec. 0.26]{FolPDE}: 
\begin{gather}
\left(f,g\right)=\left(2\pi\right)^{-n}\left(\widehat{f},\widehat{g}\right),\quad\left\Vert f\right\Vert =\left(2\pi\right)^{-n/2}\left\Vert \widehat{f}\right\Vert ,\label{eq:forier1e}\\
\left(f,g\right)\stackrel{\mathrm{def}}{=}\int_{\mathbb{R}^{n}}\overline{f\left(x\right)}g\left(x\right)\,\mathrm{d}x,\quad\left\Vert f\right\Vert \stackrel{\mathrm{def}}{=}\sqrt{\left(f,f\right)}.\nonumber 
\end{gather}

This preference was motivated by the fact that the so-defined Fourier
transform of the convolution of two functions has its simplest form.
Namely, the convolution $f\ast g$ of two functions $f$ and $g$
is defined by \citealt[Sec. 4.3.1]{Evans}, \citet[Sec. 7.2, 7.5]{Foll},
\begin{gather}
\left[f\ast g\right]\left(t\right)=\left[g\ast f\right]\left(t\right)=\int_{-\infty}^{\infty}f\left(t-t^{\prime}\right)g\left(t^{\prime}\right)\,\mathrm{d}t^{\prime},\label{eq:fourier2a}\\
\left[f\ast g\right]\left(z,t\right)=\left[g\ast f\right]\left(z,t\right)=\int_{-\infty}^{\infty}f\left(z-z^{\prime},t-t^{\prime}\right)g\left(z^{\prime},t^{\prime}\right)\,\mathrm{d}z^{\prime}\mathrm{d}t^{\prime}.\label{eq:fourier2b}
\end{gather}
Then its Fourier transform as defined by equations (\ref{eq:fouri1a})-(\ref{eq:fouri1c})
satisfies the following properties: 
\begin{gather}
\left[f\ast g\right]^{\land}\left(\omega\right)=\widehat{f}\left(\omega\right)\widehat{g}\left(\omega\right),\label{eq:fourier3a}\\
\left[f\ast g\right]^{\land}\left(k,\omega\right)=\widehat{f}\left(k,\omega\right)\widehat{g}\left(k,\omega\right).\label{eq:fourier3b}
\end{gather}

\subsection{A few facts about determinants}

\label{subsec:determ}

We present here a few important statements for determinants following
mostly \citet{ArnoODE}, \citet{BernS}, \citet{HorJohn}, \citet{PizOde}.
The theory of determinants is an important part of the linear algebra
and its geometric applications. Concepts of Grassmann exterior and
Clifford algebras give a deep insight into the properties of determinants,
\citet[Sec. 1.2, 3.3]{VeiDal}, \citet[Sec. 1.4]{HesSob}, \citet[Chap. 4]{SnyggN}.
In particular, according to \citet[Sec. 1.4]{HesSob}: 
\begin{quotation}
`` ... a determinant is nothing more nor less than the scalar product
of two blades.''. 
\end{quotation}
We introduce first basic notations. Let $M_{m,n}\left(\mathbb{F}\right)$
is a set of $m\times n$ matrices with entries in field $\mathbb{F}$.
We also use an abbreviation $M_{n}\left(\mathbb{F}\right)=M_{n,n}\left(\mathbb{F}\right)$.
To describe submatrices of a given matrix we introduce first index
sequences 
\begin{equation}
Q_{r,m}=\left\{ \left(i_{1},i_{2},\ldots,i_{r}\right)|1\leq i_{1}<i_{2}<\cdots<i_{r}\leq m\right\} ,\quad1\leq r\leq m.\label{eq:detQrm1a}
\end{equation}
Then if $A=\left\{ A_{ij}\right\} \in M_{m,n}\left(\mathbb{F}\right)$,
$\alpha\in Q_{r,m}$ and $\beta\in Q_{s,n}$ then $A\left[\alpha,\beta\right]\in M_{r,s}\left(\mathbb{F}\right)$
stands for a submatrix of $A$ with row indexes coming from $\alpha$
and column indexes coming from $\beta$. It useful to introduce also
a complimentary to $A\left[\alpha,\beta\right]\in M_{r,s}\left(\mathbb{F}\right)$
submatrix $A\left[\alpha^{c},\beta^{c}\right]\in M_{r,s}\left(\mathbb{F}\right)$
where $\alpha^{c}$ is the complimentary to $\alpha$ sequence, namely
\begin{equation}
\alpha^{c}=\left\{ 1,\ldots,m\right\} \setminus\alpha\in Q_{m-r,m},\label{eq:detQrm1b}
\end{equation}
or, in other words, $\alpha^{c}$ is obtained by removal sequence
$\alpha$ from sequence $\left\{ 1,\ldots,m\right\} $, and $\beta^{c}\in Q_{n-s,n}$
is defined similarly.

Suppose now $A=\left\{ A_{ij}\right\} \in M_{m,n}\left(\mathbb{F}\right)$
where $A_{ij}$, $1\leq i,j\leq n$ are the entries of matrix $A$.
For any pair $1\leq i,j\leq n$ we introduce a $\left(n-1\right)\times\left(n-1\right)$
submatrix $A\left[i^{c},j^{c}\right]$ obtained by deleting $i$-th
row and $j$-th column from A and refer to it \emph{cofactor} of $A_{ij}$.
We introduce also the so-called \emph{adjugate} to $A$ matrix $A^{\mathrm{A}}$
(sometimes called adjoint) defined using cofactors $A\left[i^{c},j^{c}\right]$
as follows, \citet[Sec. 3.8]{BernS}, \citet[App. C.3.3]{PizOde}:
\begin{equation}
\left[A^{\mathrm{A}}\right]_{i,j}\stackrel{\mathrm{def}}{=}\left(-1\right)^{i+j}\det\left\{ A\left[j^{c},i^{c}\right]\right\} ,\quad1\leq i,j\leq n.\label{eq:adjug1a}
\end{equation}
The adjugate matrix satisfy the following identities, \citet[Sec. 3.8, 3.19]{BernS},
\citet[App. C.3.3]{PizOde}: 
\begin{equation}
\sum_{m=1}^{n}A_{ij}\left[A^{\mathrm{A}}\right]_{m,i}=\left[AA^{\mathrm{A}}\right]_{i,i}=\left[A^{\mathrm{A}}A\right]_{i,i}=\det\left\{ A\right\} ,\label{eq:adjug1b}
\end{equation}
\begin{equation}
\sum_{m=1}^{n}A_{ij}\left[A^{\mathrm{A}}\right]_{m,j}=\left[AA^{\mathrm{A}}\right]_{i,j}=\left[A^{\mathrm{A}}A\right]_{i,j}=0,\quad i\neq j,\label{eq:adjug1c}
\end{equation}
\begin{equation}
A^{\mathrm{A}}A=AA^{\mathrm{A}}=\det\left\{ A\right\} \mathbb{I},\label{eq:adjug1d}
\end{equation}
where $\mathbb{I}$ is the identity matrix. Note in case when $A$
is not degenerate the identity (\ref{eq:adjug1d}) readily implies
the following representation: 
\begin{equation}
A^{\mathrm{A}}=\frac{1}{\det\left\{ A\right\} }A^{-1},\quad\det\left\{ A\right\} \neq0.\label{eq:adjug1e}
\end{equation}

\subsubsection{Laplace expansion}

The Laplace expansion represents the determinant of a square matrix
in terms of the product of determinants of certain submatrices. Here
is its main statement, \citet[Sec. 3.3]{VeiDal}, \citet[Sec. 1.4]{HesSob},
\citet[Sec. 0.8.9]{HorJohn}, \citet[App. C.3.3]{PizOde}. 
\begin{thm}[Laplace expansion theorem]
\label{thm:Lapexp}Let $A\in M_{n}\left(\mathbb{F}\right)$ and $\alpha\in Q_{r,n}$
for $1\leq r\leq n$ be fixed. Then the following Laplace expansion
of $\det\left\{ A\right\} $ by rows holds 
\begin{equation}
\det\left\{ A\right\} =\sum_{\beta\in Q_{r,n}}\left(-1\right)^{\left|\alpha\right|+\left|\beta\right|}\det\left\{ A\left[\alpha|\beta\right]\right\} \det\left\{ A\left[\alpha^{c},\beta^{c}\right]\right\} ,\label{eq:detQrm2a}
\end{equation}
where for $\alpha=\left(i_{1},i_{2},\ldots,i_{r}\right)$ quantity
$\left|\alpha\right|$ is defined by 
\begin{equation}
\left|\alpha\right|=\left|\left(i_{1},i_{2},\ldots,i_{r}\right)\right|=i_{1}+i_{2}+\cdots+i_{r}.\label{eq:detQrm2b}
\end{equation}
Similarly, if $\beta\in Q_{r,n}$ for $1\leq r\leq n$ is fixed, then
the following Laplace expansion of $\det\left\{ A\right\} $ by columns
holds 
\begin{equation}
\det\left\{ A\right\} =\sum_{\alpha\in Q_{r,n}}\left(-1\right)^{\left|\alpha\right|+\left|\beta\right|}\det\left\{ A\left[\alpha|\beta\right]\right\} \det\left\{ A\left[\alpha^{c},\beta^{c}\right]\right\} .\label{eq:detQrm2c}
\end{equation}
\end{thm}

\subsubsection{The Liouville-Jacobi formula}

Suppose that $A(t)$ and $M(t)$ is $n\times n$ matrices satisfying
the following Cauchy problem 
\begin{equation}
\frac{dM(t)}{dt}=A(t)M(t),\quad M(0)=\mathbb{I},\label{eq:LiJac1a}
\end{equation}
where $A(t)$ is $n\times n$ matrix. Then the following \emph{Liouville-Jacobi
formula} holds, \citet[Sec. 27.6]{ArnoODE}, \citet[Sec. XIV.1]{GantMa2},
\citet[Sec. II.1.2]{YakSta}: 
\begin{equation}
\det\left\{ M(t)\right\} =\exp\left[\intop_{0}^{t}\mathrm{Tr}\,\left\{ A(\tau)\right\} \,\mathrm{d}\tau\right].\label{eq:LiJac1b}
\end{equation}

In the case when $A(t)=A$ is a constant matrix the Liouville-Jacobi
formula readily yields the following identity, \citet[Sec. 16.3, 16.4, 27.6]{ArnoODE},
\citet[Section 15.2]{BernS} 
\begin{equation}
\det\left\{ \exp\left[A\right]\right\} =\exp\left[A\right].\label{eq:LiJac1c}
\end{equation}

The following statement holds, \citet[Sec. XIII.16 Lemma 6]{ReeSim4} 
\begin{lem}
For any matrix $A=\left\{ a_{ij}\right\} $ and $\tau$ small 
\[
\det\left\{ \mathbb{I}+\tau A\right\} =\exp\left\{ -\sum_{m=1}^{\infty}\tau^{k}\frac{\mathrm{Tr\,}\left(-A\right)^{k}}{k}\right\} ,
\]
where $\mathbb{I}$ is the identity matrix and $\mathrm{Tr\,}\left(A\right)=\sum_{i=1}^{n}a_{ii}$
is the trace of matrix $A$, that is the sum of its diagonal entries.

In particular, 
\[
\det\left\{ \mathbb{I}+\tau A\right\} =1+\tau\mathrm{Tr\,}\left(A\right)+O\left(\tau^{2}\right),\quad\tau\rightarrow0,
\]
. 
\end{lem}

\subsubsection{Determinant of the sum of two matrices}

\label{subsec:detApB}

Using notations for matrices and submatrices introduced in the beginning
of Section \ref{subsec:determ} we write the following formula for
the determinant of the sum of two $n\times n$ matrices with $n\geq2$
which due to Markus, \citet{Markus} 
\begin{equation}
\det\left\{ A+B\right\} =\det\left\{ A\right\} +\det\left\{ B\right\} +\sum_{r=1}^{n-1}\sum_{\alpha,\beta\in Q_{r,n}}\left(-1\right)^{\left|\alpha\right|+\left|\beta\right|}\det\left\{ A\left[\alpha|\beta\right]\right\} \det\left\{ B\left[\alpha^{c},\beta^{c}\right]\right\} ,\label{eq:detApB1}
\end{equation}

\textbf{\vspace{0.2cm}
 }

\textbf{DATA AVAILABILITY:} The data that support the findings of
this study are available within the article.\textbf{\vspace{0.2cm}
 }

\end{document}